\documentclass{sig-alternate}
\usepackage{balance}  
\usepackage{alex}
\usepackage{graphicx}
\usepackage{subfigure}
\usepackage[colorlinks]{hyperref}

\usepackage{algorithm}
\usepackage{algorithmic}
\usepackage{multirow}

\usepackage{color}
\newcommand{\ft}[1]{{\color{red}{\underline{#1}}}}
\newcommand{\sd}[1]{{\color{green}{#1}}}

\usepackage{textcomp}
\usepackage{xspace}
\def\algo{Parsa\xspace}

\newcommand{\ps}{parameter server\xspace}

\def\ctrb{\textsf{CTRb}\xspace}

\begin{document}

\title{Graph Partitioning via Parallel Submodular Approximation\\
  to Accelerate Distributed Machine Learning}

\numberofauthors{3}
\author{
\alignauthor
Mu Li\\\vspace{.5ex}
\affaddr{Carnegie Mellon University}\\
\email{muli@cs.cmu.edu}
\alignauthor
Dave G.\ Andersen\\\vspace{.5ex}
\affaddr{Carnegie Mellon University}\\
\email{dga@cs.cmu.edu}
\alignauthor
Alexander J.\ Smola\\\vspace{.5ex}
\affaddr{Carnegie Mellon University}\\
\email{alex@smola.org}
}

\newcommand{\aff}[1]{\raisebox{.6ex}{\small\textsf{#1}}}
\newcommand{\fa}{\aff{1}}
\newcommand{\fb}{\aff{2}}
\newcommand{\fc}{\aff{3}}
\newcommand{\itv}{\hspace{1ex}}

\maketitle

\begin{abstract}
  Distributed computing excels at processing large scale data, but the communication cost
  for synchronizing the shared parameters may slow down the overall
  performance. Fortunately, the interactions between parameter and data in many problems
  are sparse, which admits efficient partition in order to
  reduce the communication overhead.

  In this paper, we formulate data placement as a graph partitioning problem. We
  propose a distributed partitioning algorithm.  We give both theoretical
  guarantees and a highly efficient implementation.  We also provide a highly
  efficient implementation of the algorithm and demonstrate its promising
  results on both text datasets and social networks.  We show that the proposed
  algorithm leads to 1.6x speedup of a state-of-the-start distributed machine
  learning system by eliminating 90\% of the network communication.
\end{abstract}

\section{Introduction}
\label{sec:intro}

The importance of large-scale machine learning continues to grow in concert with
the big data boom, the advances in learning techniques, and the deployment of
systems that enable wider applications. As the amount of data scales up, the
need to harness increasingly large clusters of machines significantly
increases. In this paper, we address a question that is fundamental for applying
today's loosely-coupled ``scale-out'' cluster computing techniques to important
classes of machine learning applications:

\begin{quote}
{How to spread data and model parameters across a cluster of
machines for efficient processing?}
\end{quote}

\begin{figure}[t!]
  \centering
  \noindent%
  \includegraphics[width=\columnwidth]{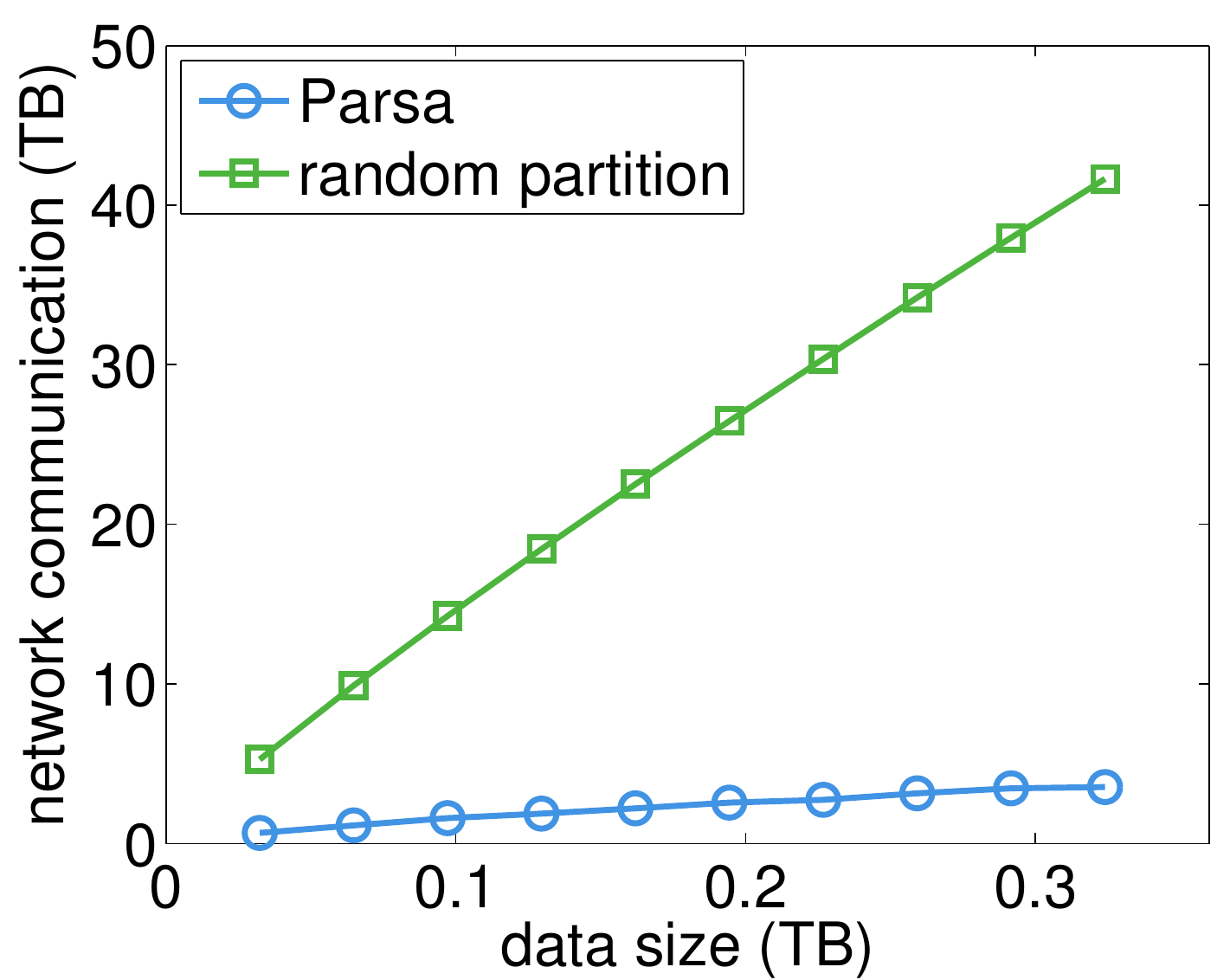}
  \vspace{-1mm}
  \caption{The amount of (outgoing) network traffic versus the size of data in a real text
    dataset.  The algorithm uses 16 machines to run 100 iterations.  The first-order
    gradients are communicated. }
    \label{fig:traffic_vs_datasize}
  \vspace{-3mm}
\end{figure}


One big challenge for large-scale data processing problems is to distribute the data over
processing nodes to fit the computation and storage capacity of each node.
%
For instance, for very large scale graph factorization~\cite{AhmSheNarJosSmo13}, one needs
to partition a natural graph in a way such that the memory, which is required for storing
local state of the partition and caching the adjacent variables, is bounded within the
capacity of each machine.
%
Similar constraints apply to GraphLab~\cite{LowGonKyrBicetal10,GonLowGuBicetal12}, where
vertex-specific updates are carried out while keeping other variables synchronized between
machines.
Likewise, in distributed inference for graphical models with latent
variables~\cite{AhmAlyGonNaretal12,SmoNar10}, the distributed state variables must be
synchronized efficiently between machines.
Furthermore, general purposed distributed machine learning framework such as the
parameter server~\cite{LiAndParSmoetal14, DeaCorMonCheetal12} face similar
issues when it comes to data and parameter layout.

Shared parameters are synchronized via the communication network. The
sheer number of parameters and the iterative nature of machine learning algorithms often
produce huge amounts of network traffic. Figure~\ref{fig:traffic_vs_datasize} shows that,
if we randomly assign data (documents) to machines in a text classification application,
the total amount of network traffic is 100 times larger than the size of training
data. Specifically, almost 4 TB parameters are communicated for 300 GB training
data. Given that the network bandwidth is typically much smaller than the local memory bandwidth,
this traffic volume can potentially become a performance bottleneck.

There are three key challenges in achieving scalability for large-scale data processing
problems:
\begin{description}
\item[Limited computation (CPU) per machine:] therefore we need a well-balanced task
  distribution over machines.
\item[Limited memory (RAM) per machine:] the
  amount of storage per machine available for processing and
  caching model variables is often constrained to a small
  fraction of the total model.
\item[Limited network bandwidth:] the network bandwidth is typically 100 times worse than
  the local memory bandwidth. Thus we need to reduce the amount of communication between
  machines.
\end{description}
One key observation is the sparsity pattern in large scale datasets: most documents
contain only a small fraction of distinct words, and most people have only a few friends
in a social graph. Such nonuniformity and sparsity is both a boon and a challenge for the
problem of dataset partitioning.  Due to its practical importance, even though the dataset
partitioning problems are often NP hard~\cite{StaKli12}, it is still worth seeking
practical solutions that outperform random partitioning, which typically leads to poor
performance.


{\bfseries Our contributions:}
In this paper, we formulate the task of data and parameter placement as a graph
partitioning problem.  We propose \emph{\algo}, a PARallel Submodular Approximation
algorithm for solving this problem, and we analyze its theoretical guarantees.  A
straightforward implementation of the algorithm has running time in the order of $\Ocal(k
|E|^2)$, where $k$ is the number of partitions and $|E|$ is the number of edges in the
graph.  Using an efficient vertex selection data structure, we provide an
efficient implementation with time complexity $\Ocal(k |E|)$.  We also
discuss the techniques including sampling, initialization and parallelization to improve
the partitioning quality and efficiency.

Experiments on text datasets and social networks of various scales show that, on both
partition quality and time efficiency, \algo outperforms state-of-the-art methods,
including METIS \cite{KarKum98}, PaToH \cite{CatAyk99} and Zoltan \cite{Devineetal06}.
\algo can also significantly accelerate the \ps, a state-of-the-art general purpose
distributed machine learning, on data of hundreds GBs size and with billions
parameters.


\section{Graph Partitioning}
\label{sec:graphpart}

In this section, we first introduce the inference problem and the model of dependencies in
distributed inference.  Then we provide the formulation of the data partitioning problem
in distributed inference. We also present a brief overview of related work in the end.

\subsection{Inference in Machine Learning}

In machine learning, many inference problems have graph-structured dependencies. For
instance, in risk minimization~\cite{HasTibFri09}, we strive to solve
\begin{align}
  \label{eq:reg-risk}
  \mini_w\ R[w] := \sum_{i=1}^m l(x_i, y_i, w) + \Omega[w],
\end{align}
where $l(x_i, y_i, w)$ is a loss function measuring the model fitting error in the data
$(x_i,y_i)$, and $\Omega[w]$ is a regularizer on the model parameter $w$. The data and
parameters are often correlated only via the nonzero terms in $x_i$, which exhibit
sparsity patterns in many applications.  For example, in email spam filtering, elements of
$x_i$' correspond to words and attributes in emails, while in computational advertising,
they correspond to words in ads and user behavior patterns.

For undirected graphical models \cite{Besag74,KscFreLoe01}, the joint distribution of the
random variables in logscale can be written as a summation of potential of all the cliques
in the graph, and each clique potential $\psi_C(w_C)$ only depends on the subset of
variables $w_C$ in the clique $C$.

The learning and inference problems in undirected graphical models are often formulated as
an optimization problem in the following form:
\begin{align}
  \label{eq:cliquerisk}
  \mini_w\ R[w] := \sum_{C \in \Ccal} \psi_C(w_C),
\end{align}
where local variables interact through the model parameters $w_C$ of the cliques.

Similar problems occur in the context of inference on natural graphs
\cite{AndGleMir12,GonLowGuBicetal12,AhmSheNarJosSmo13}, where we have sets of interacting
parameters represented by vertices on the graph, and manipulating a vertex affects all of
its neighbors computationally.

\begin{figure}[t!]
  \centering
  \includegraphics[width=\columnwidth]{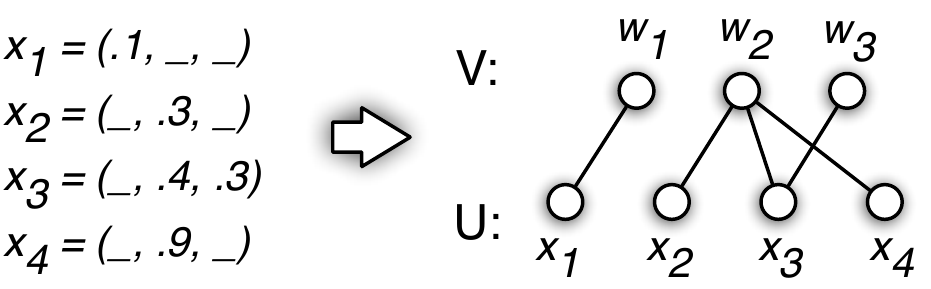} \vspace{-2mm}
  \caption{Modeling dependences as bipartite graph
    \label{fig:dep}}
  \vspace{-4mm}
\end{figure}

\subsection{Bipartite Graphs}

The dependencies in the inference problems above can be modeled by a bipartite graph
$G(U,V,E)$ with vertex sets $U$ and $V$ and edge set $E$.  We denote the edge between two
node $u\in U$ and $v\in V$ by $(u,v)\in E$.
Figure~\ref{fig:dep} illustrates the case of risk minimization \eq{eq:reg-risk}, where $U$
consists of the samples $\cbr{(x_i,y_i)}_{i=1}^m$ and $V$ consists of the parameters in
$w$.  There is an edge $((x_i,y_i), w_j)$ if and only if the $j$-th element of $x_i$ is
non-zero. Therefore, $\cbr{w_j: ((x_i,y_i),w_j)\in E}$ is the working set of elements of
$w$ for evaluating the loss function $l(x_i,y_i,w)$ on the sample $(x_i,y_i)$.

We can construct such bipartite graph $G(U',V,E')$ to encode the dependencies in
undirected graphical models and natural graphs with node set $V$ and edge set $E$.  One
construction is to define $U'=V$, and add an edge $(u, v)$ to the edge set $E'$ if they
are connected in the original graph. An alternative construction is to define the node set
$U'$ to be $\mathcal C$, the set of all cliques of the original graph, and add an edge $(C, v)$
to the edge set $E'$ if node $v$ belongs to the clique $C$ in the original graph.

Throughout the discussion, we refer to $U$ as the set of data (examples) nodes and $V$ as
the set of parameters (results) nodes.



\begin{figure}[t!]
  \centering
  \includegraphics[width=\columnwidth]{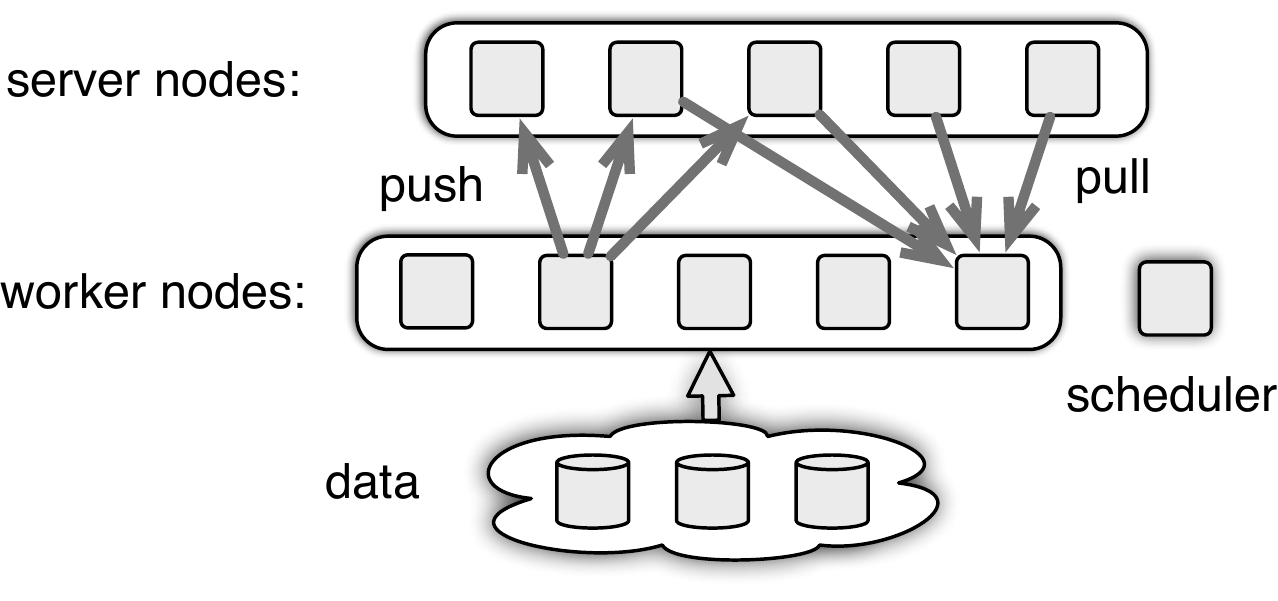}
  \vspace{-1mm}
  \caption{Simplified parameter server architecture.}
  \label{fig:ps_arch}
  \vspace{-3mm}
\end{figure}

\subsection{Distributed Inference}

The challenge for large scale inference is that the size of the optimization problem in
\eqref{eq:cliquerisk} is too large, and even the model $w$ may be too large to be stored
on a single machine.
One solution is to divide exploit the additive form of $R[w]$ to decompose the
optimization into smaller problems, and then employ multiple machines to solve these
sub-problems while keeping the solutions (parameters) consistent.

There exist several frameworks to simplify the developing of efficient distributed
algorithms, such as Hadoop~\cite{Hadoop} and its in-memory variant
Spark~\cite{ZahChoDasDavetal12a} to execute MapReduce programs, and Graphlab for
distributed graph computation~\cite{LowGonKyrBicetal12}. In this paper, we focus on the
\ps framework~\cite{LiAndParSmoetal14}, a high-performance general-purpose distributed
machine learning framework.

In the \ps framework, computational nodes are divided into server nodes and worker nodes,
which are shown in Figure~\ref{fig:ps_arch}.  The globally shared parameters $w$ are
partitioned and stored in the server nodes.  Each worker node solves a sub-problem and
communicates with the server nodes in two ways: to push local results such as gradients or
parameter updates to the servers, and to pull recent parameter (changes) from the
servers. Both { push} and { pull} are executed asynchronously.


\subsection{Multiple Objectives of Partitioning }

\begin{figure}[t!]
  \centering
  \includegraphics[width=\columnwidth]{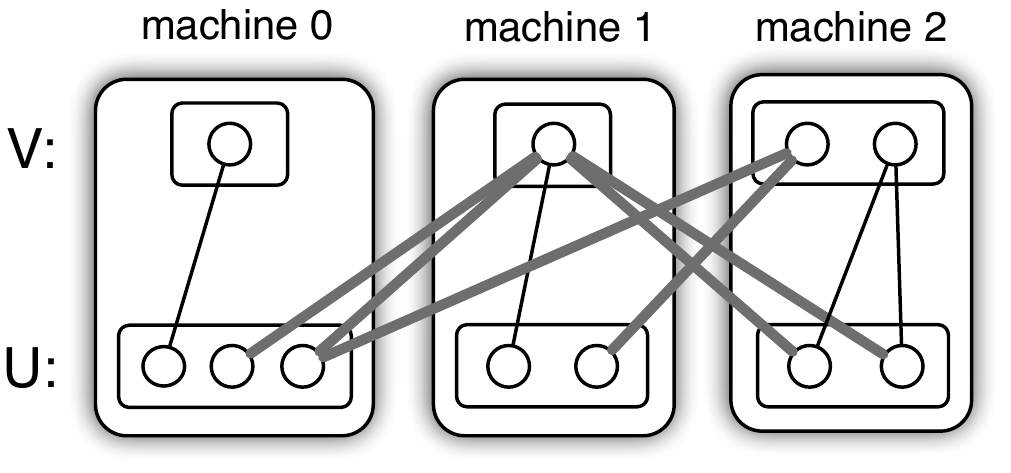}
  \vspace{-2mm}
  \caption{Each machine contains a server and a worker, holding a part
    of $U$ and $V$, respectively. The inter-machine dependencies (edges) are
    highlighted and the communication costs for these three machines are 2, 3,
    and 3, respectively. Moving the 3rd vertex in $V$ to
    either machine 0 or 1 reduces cost.
    \label{fig:ps_partition}}
  \vspace{-4mm}
\end{figure}

In distributed inference, we divide the problem in \eqref{eq:cliquerisk} by partition the
cost function $R[w]$ as well as the associated dependency graph into $k$ blocks.
Without loss of generality we consider a \ps with $k$ server nodes and $k$ worker nodes,
and each machine has exactly one server and one worker (otherwise we can aggregate
multiple nodes in the same machines without affecting the following analysis). For the
bipartite dependency graph $G(U,V,E)$, we partition the parameter set $V$ into $k$ parts
and assign each part to a server node, and we partition the data set $U$ into $k$ parts
and assign them to individual worker nodes. Figure~\ref{fig:ps_partition} illustrates an
example for $k=3$.
More specifically, we want to divide both $U$ and $V$ into $k$ non-overlapping parts
\begin{equation}
  U = \bigcup_{i=1}^k U_i \quad \textrm{and} \quad V = \bigcup_{i=1}^k V_i,
\end{equation}
and assign the part $U_i$ and $V_i$ to the worker node and server node on machine $i$
respectively.

There are three goals when implementing the graph partitioning:

\noindent {\bf Balancing the computational load.} We want to ensure that each machine has
approximately the same computational load. Assume that each example $u_i$ incurs roughly
the same workload, then one of the objective to keep $\max_i |U_i|$ small:
\begin{equation}
  \label{eq:cpu}
  \mini ~ \max_i |U_i|
\end{equation}
{\bf Satisfying the memory constraint.} Inference algorithms
frequently access the parameters (at random). Workers keep these
parameters in memory to improve performance, yet RAM is limited.
Denote by $\Ncal(u_i)$ the neighbor set of $u_i$
\begin{equation}
  \Ncal(u_i) = \cbr{v_j : (u_i, v_j) \in E}.
\end{equation}
Then $\bigcup_{u \in U_i} \Ncal(u)$
is the working set of the parameters worker $i$ needed. For simplicity
we assume that each parameter
$v_j$ has the same storage cost. Our goal to
limit the worker's memory footprint is given by
\begin{equation}
  \label{eq:mem}
 \mini ~ \max_i |\Ncal(U_i)| ~ \textrm{ where } \Ncal(U_i) := \bigcup_{u \in
   U_i} \Ncal(u)
\end{equation}
{\bf Minimizing the communication cost.}
The total communication cost per worker $i$ is $|\Ncal(U_i)|$, which
is already minimized using our previous goal \eq{eq:mem}. To further reduce this
cost, we can assign server $i$ to the same machine with worker $i$, so
that any communication uses memory rather than network. This reduces
the inter-machine communication cost to
$|\Ncal(U_i)| - |\Ncal(U_i) \backslash V_i|$. Figure~\ref{fig:ps_partition} shows an example. Further note
that if $v_j$ is not needed by worker $i$, then server $i$ should
never maintain
$v_j$. In other words, we have $V_i \subseteq \Ncal(U_i)$ and the cost
simplifies to $|\Ncal(U_i)| - |V_i|$.

On the other hand, $\sum_{j \neq i} \abr{V_i \cap \Ncal(U_j)}$ is the
communication cost of server $i$ because other workers must request
parameters from server $i$.
Therefore, the goal to minimize the maximal communication cost of a machine is
\begin{align}
  \label{eq:commcost}
  \mini ~ \max_i {\abr{\Ncal(U_i)} - |V_i| + \sum_{j \neq i} \abr{V_i \cap
      \Ncal(U_j)}}.
\end{align}

\subsection{Related Work}

Graph partitioning has attracted much interest in scientific
 computing \cite{KarKum98, CatAyk99, Devineetal06}, scaling out
large-scale computations \cite{GonLowGuBicetal12, ZhoBruLin12, YanYanZonKha12,
  BouLelVoj14, GonXinDavCraFraetal14, UgaBac13}, graph databases
\cite{ShaWanLi13, VenAmsBroCabChaetal12, CurBecBosDorGrietal13}, search and
social network analysis \cite{PujErrSigYanLaoChhetal11,
  UgaBac13,AhmSheNarJosSmo13}, and streaming processing \cite{StaKli12,
  Stanton12, NisUga13,TsoGkaRadVoj14}.

Most previous work, such as METIS \cite{KarKum98}, is
concerned with edge cuts. Only a few of them solve the vertex cut problem,
which is closely related to this paper, to directly minimize the network
traffic. PaToH \cite{CatAyk99} and Zoltan \cite{Devineetal06} used multilevel
partitioning algorithms related to METIS, while PowerGraph
\cite{GonLowGuBicetal12} adopted a greedy algorithm. Very recently
\cite{BouLelVoj14} studied the relation between edge cut and vertex cut.

Different to these works, we propose a new algorithm based on submodular
approximation to solve the vertex-cut partitioning problem. We give theoretical
analysis of the partition quality, and describe an efficient distributed
implementation. We show that the proposed algorithm outperforms the
state of the art on several large scale real datasets in both in terms
of quality and speed.

\section{Algorithm}
\label{sec:parallel}

In this section, we present our algorithm Parsa for solving the partitioning problem with
multiple objectives  in \eq{eq:cpu}, \eq{eq:mem} and \eq{eq:commcost}.

Note that \eq{eq:mem} is equal to a $k$-way graph partition problem on vertex set $U$ with
vertex-cut as the merit. This problem is NP-Complete \cite{CatAyk99}. Furthermore,
\eq{eq:commcost} is more complex because of the involvement of $V$.
Rather than solving all these objectives together, Parsa decomposes this problem into two
tasks: partition the data $U$ by solving \eq{eq:cpu} and \eq{eq:mem}, and given the
partition of $U$ partition the parameters $V$ by solving \eq{eq:commcost}.
Intuitively, we first assigns data workers to balance the CPU load and minimize the memory
footprint, and then distribute the parameters over servers to minimize inter-machine
communication.

\subsection{Partitioning $U$ over Worker Nodes}
\label{sec:partition-U}

\begin{algorithm}[tb]
  \caption{Partition $U$ via submodular approximation}
  \label{algo:prototype}
  \begin{algorithmic}[1]
    \REQUIRE Graph $G$, \#partitions $k$, maximal \#iterations
    $n$, residue $\theta$, and improvement $\alpha$
    \ENSURE Partitions of $U=\bigcup_{i=1}^k U_i$
    \FOR {$i = 1, \ldots, k$}
    \STATE $U_i\gets \emptyset$
    \STATE define $g_i(T) := f(T \cup U_i) - \alpha|T \cup U_i|$
    \ENDFOR
    \FOR {$t = 1,...,n$}
    \STATE {\bfseries if} $|U|\le k \theta$ {\bfseries then break}
    \STATE \label{algo:select_R} find $i \gets \argmin_j |U_j|$
    \STATE draw $R\subseteq U$  by choosing $u \in U$ with
    probability $\frac{n}{|U| k}$ \!\!\!
    \STATE {\bfseries if} $|R| > 2n/k$ {\bfseries then next}
    \STATE \label{algo:solv_T} solve
    $ T^* = \argmin_{T \subseteq R} g_i(T)$
    \STATE \label{algo:updt_si} {\bfseries if} $g_i(T^*)\le 0$ {\bfseries then} $U_i \gets U_i \cup T^*$ and
    $U \gets U \setminus T^*$
  \ENDFOR
  \STATE {\bfseries if} {$|U| > k \theta$} {\bfseries then return} fail
  \STATE evenly assign the remainder $U$ to $U_i$
\end{algorithmic}
\end{algorithm}

Note that $f(U) := \abr{\Ncal(U)}$ is a set function in the variable $U$.  It is a {\em
  submodular} function similar to convex and concave functions in real variables. Although
the problem in \eqref{eq:cpu} is NP-Complete, there exist several algorithms to solve it
approximately by exploiting the submodularity \cite{SviFle11}.
In our algorithm, we modified \cite{SviFle11} to solve \eq{eq:cpu} and \eq{eq:mem}.  The key difference is that we build up the sets
$U_i$ incrementally, which is important for both partition quality and computational
efficiency at a later stage.

As shown in Algorithm~\ref{algo:prototype}, the algorithm proceeds as follows: in each
round we pick the smallest partition $U_i$ and find the best set of elements to add to
it. To do so, we first draw a small subset of candidates $R$ and select the best subset
using a minimum-increment weight via $\min_{T\subseteq R}f(U_i \cup T) - \alpha |U_i\cup
T|$. If the optimal solution $T^*$ satisfies $f(U_i \cup T^*) < \alpha |U_i\cup T^*|$,
i.e., the cost for increasing $U_i$ is not too large, we assign $T^*$ to partition $U_i$.

Before showing the implementation details in
Section~\ref{sec:impl}, we first analyze the partitioning quality of Algorithm~\ref{algo:prototype}.

\begin{proposition}
  \label{prop:bound}
  Assume that there exists some partitioning $U_i^*$ that satisfies
  $\max_i f(U_i^*) \leq B$. Let $k > \theta = \sqrt{n/\log n}$,
  $c = \rbr{32 \pi}^{-\frac{1}{2}}$,
  $\alpha = BK/\sqrt{n \log n}$ and
  $\tau =  \frac{n^3}{c}\log \frac{1}{1-p}$.
  Then Algorithm~\ref{algo:prototype} will succeed with probability at least $p$
  and it will generate a feasible solution with partitioning cost at most
  $\max_i f(U_i) \leq 4B \sqrt{n/\log n}$.
\end{proposition}
\begin{proof}
  The proof is near-identical as \cite{SviFle11}. Note that we overload the
  meaning of $U$ as it refers to the remaining variables in the algorithm.

  For a given iteration, without loss of generality we assume that $U_1^*$ maximizes
  $|U_j^*\cap U|$ for all $j$. Denote this by $\upsilon =|U_1^* \cap U|$. Since $T^*$ is
  the optimal solution at the current iteration we have $g_i(T^*) \leq g_i(U_1^* \cap U) =
  f((U_1^* \cap U) \cup U_i) - \alpha \abr{(U_1^* \cap U) \cup U_i}$. Further note that by
  monotonicity and submodularity $f((U_1^* \cap U) \cup U_i) \leq f(U_1^* \cap U) +
  f(U_i)$. Moreover, $U \cap U_i = \emptyset$ holds since $U$ contains only the
  leftovers. Consequently $\abr{(U_1^* \cap U) \cup U_i} = \abr{U_1^* \cap U} +
  \abr{U_i}$.  Finally, the algorithm only increases the size of $U_i$ whenever the cost
  is balanced. Hence $f(U_i) - \alpha |U_i| < 0$. Combining this yields
  \begin{align*}
    g_i(T^*) \leq g_i(U_1^* \cap U) \leq
    f(U_1^* \cap U) - \alpha |U_1^* \cap U| \leq B - \alpha \upsilon
  \end{align*}
  Using the results from the proof of \cite[Theorem 5.4]{SviFle11} we
  know that $\upsilon \ge \frac{B}{\alpha}$ and therefore $g_i(T^*)$
  happens with probability at least $\frac{c}{n^2}$.  Hence the
  probability of removing at least one vertex from $U$ within an
  iteration is greater than $\frac{c}{n^2}$.

  Chernoff bounds show that after
  $\tau = -n^2/c \log (1-\delta)$ iterations the algorithm will
  terminate with probability at least $p$ since the residual $U$ is
  small, i.e.,\ $|U| \leq k\theta $.

  The algorithm will never select a $U_i$
  for augmentation unless $|U_i| \leq n/k$ (there would always be a
  smaller set). Moreover, the maximum increment at any given time is
  $2n/k$. Hence $|U_i| \leq 3n/k$ and therefore $f(U_i) \leq 3 n
  \alpha/k$.

  Finally, the contribution of the unassigned residual $U$ is at most
  $\theta B$ since each $U_i$ is incremented by at most $\theta$
  elements and since $f(u) \leq B$ for all $u \in U$. In summary, this
  yields
  $f(U_i) \leq 3 n \alpha/k + B \theta = 4B \sqrt{n/\log n}$.
\end{proof}

\subsection{Partitioning $V$ over Server Nodes}

\begin{algorithm}[t!]
  \caption{Partition $V$ for given $\cbr{U_i}_{i=1}^k$}
  \label{algo:part_V}
  \begin{algorithmic}[1]
    \REQUIRE The neighbor sets $\{\Ncal(U_i)\}_{i=1}^k$
    \ENSURE Partitions $V=\bigcup_{i=1}^k V_i$
    \FOR {$i=1,\ldots k$}
    \STATE $V_i\gets \emptyset$
    \STATE $\textrm{cost}_i\gets |\Ncal(U_i)|$
    \ENDFOR
    \FOR {{\bfseries all } $j \in V$}
    \STATE $ \xi \gets \argmin_{i: u_{ij} \neq 0} \textrm{cost}_i$
    \STATE $V_{\xi} \gets V_{\xi} \cup \{j\}$
    \STATE $\textrm{cost}_{\xi} \gets \textrm{cost}_{\xi} - 1 + \sum_{i \neq \xi} u_{ij}$
    \ENDFOR
  \end{algorithmic}
\end{algorithm}


Next, given the partition of $U$, we find an assignment of parameters in $V$ to servers.
We reformulate \eq{eq:commcost} as a convex integer programming problem with totally
unimodular constraints \cite{HelTom56}, which is then solved using a sequential
optimization algorithm performing a sweep through the variables.

We define index variables $v_{ij} \in \cbr{0, 1}$, $j=1,\ldots,k$ to indicate
which server node maintains a particular parameter $v_i$. They need to
satisfy $\sum_{j=1}^k v_{ij} = 1$. Moreover, denote by $u_{ij} \in
\cbr{0, 1}$ variables that record whether $j \in \Ncal(U_i)$. Then we
can rewrite \eq{eq:commcost} as a convex integer program:
\begin{subequations}
  \label{eq:intprog}
\begin{align}
  \mini_{v} ~ & \max_i |\Ncal(U_i)| + \sum_j v_{ij} \sbr{-1 + \sum_{l
      \neq i} u_{lj}} \\
  \text{subject to} ~ & \sum_{j} v_{ij} = 1
  \text{ and } v_{ij} \in \cbr{0, 1}
  \text{ and } v_{ij} \leq u_{ij}
\end{align}
\end{subequations}
Here we exploited the fact that
$\sum_j v_{ij} u_{lj} = |V_i \cap \Ncal(U_l)|$ and that
$\sum_{j=1}^k v_{ij} = |V_i|$. These constraints are totally
unimodular, since they satisfy the conditions of \cite{HelTom56}. As a
consequence every vertex solution is integral and we may relax the
condition $v_{ij} \in \cbr{0, 1}$ to $v_{ij} \in [0, 1]$ to obtain a
convex optimization problem.


Algorithm~\ref{algo:part_V} performs a single sweep over
\eq{eq:intprog} to find a locally optimal assignment of one variable
at a time. We found that it is sufficient for a near-optimal
solution. Repeated sweeps over the assignment space are
straightforward and will improve the objective until convergence to
optimality in a \emph{finite} number of steps: due to convexity all
local optima are global. Further note that we need not store the full
neighbor sets in memory. Instead, we can perform the assignment in a
streaming fashion.

\section{Efficient Implementation}
\label{sec:impl}

The time complexity of Algorithm~\ref{algo:part_V} is $\Ocal(k(|U|+|V|))$,
however, it could be $\Ocal(k|U|^6)$ for
Algorithm~\ref{algo:prototype}, which is infeasible in practice.
We now discuss how to implement Algorithm~\ref{algo:prototype}
efficiently. We first present how to find the optimal $T^*$ and
sample $R$. Then we address the parallel implementation with the \ps,
and finally describe neighbor set initialization to improve the
partition quality.

\subsection{Finding ${T}^*$  efficiently}
\def\top{\textrm{min}}

\begin{algorithm}[tb]
  \caption{Partitioning $U$ efficiently}
  \label{algo:assign_a_blk}
  \begin{algorithmic}[1]
    \REQUIRE Graph $G(U,V,E)$, \#partitions $k$, and initial neighbor sets $\{S_i\}_{i=1}^k$
    \ENSURE Partitioned $U=\bigcup_{i=1}^k U_i$ and updated neighbor sets which
    are equal to $\{S_i\cup \Ncal(U_i)\}_{i=1}^k$
     \FOR {$i=1,\ldots,k$}
     \STATE $U_i\gets \emptyset$
     \STATE {\bf for all} $u \in U$ {\bf do} $\Acal_i(u) = |\,\Ncal(u)\setminus S_i\,|$
    \ENDFOR
    \WHILE {$|U|>0$}
    \STATE pick partition $i\gets \arg\min_j |S_j|$
    \STATE \label{step:top} pick the lowest-cost vertex $u^* \gets \Acal_i.\top$
    \STATE assign $u^*$ to partition $i$: $U_i\gets U_i \cup  \{u^*\}$
    \STATE remove $u^*$ from $U$: $U\gets U\setminus \{u^*\}$
    \STATE \label{step:remove} {\bf for} $j = 1, \ldots, k$ {\bf do} remove $u^*$ from $\Acal_j$
    \FOR {$v \in \Ncal(u^*)\setminus  S_i$}
    \STATE $S_i \gets S_i \cup \{v\}$
    \STATE  \label{step:cost} \label{step:update} {\bf for} $u \in \Ncal(v) \cap
    U$ {\bf do} $\Acal_i(u) \gets \Acal_i(u)-1$
    \ENDFOR
    \ENDWHILE
  \end{algorithmic}
\end{algorithm}

\begin{figure}[tb]
  \centering
  \includegraphics[width=\columnwidth]{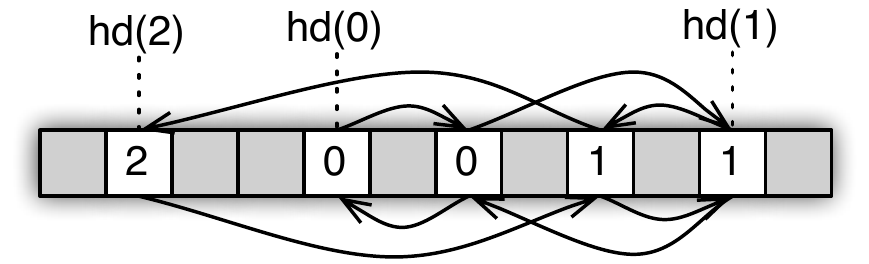}
  \vspace{-2mm}
  \caption{Vertex costs are stored in an array. The $i$-th entry is
    used for vertex $u_i$, where assigned vertices are
    marked with gray color. Header points and
    a doubly-linked list afford faster access.
  \label{fig:array}}
  \vspace{-2mm}
\end{figure}

The most expensive operation in the inner loop of Algorithm~\ref{algo:prototype}
is step~\ref{algo:solv_T}, determining which vertices, $T^*$, to add to a
partition.  Submodular minimization problems incur $\Ocal(n^6)$ time
\cite{Orlin09}. Given the fact that this step is invoked frequently and the
problem is large, this strategy is impractical.
A key approximation Parsa made is to add only a single vertex at a time
instead of a set of vertices:
Given a vertex set $R$ and partition $i$, it finds vertex $u^*$ that minimizes
\begin{equation}
  \label{eq:T*}
  u^*=\argmin_{u\in R} g_i(u) := |\Ncal(\cbr{u}\cup U_i)| - \alpha |\!\cbr{u}\cup U_i|
\end{equation}
An additional advantage of this approximation is that we are now
solving exactly the CPU load balancing problem \eq{eq:cpu}. Since we
only assign one vertex at a time to the smallest partition, we obtain
perfect balancing.

Even though this approximation improves the performance, a naive way to
calculate \eq{eq:T*} is to compute all $g_i(u)$ to find $u^*$ with the
minimal value for each iteration. If the size of $R$ is a constant fraction of
the entire graph, this leads to an undesirable time complexity of
$\Ocal(|U| |E|)$.  This remains impractical for graphs with billions of
vertices and edges.

We accelerate computation as follows: we store all vertex costs
to avoid re-computing them, and we create a data structure
to locate the lowest-cost vertex efficiently.

{\bfseries Storing vertex costs.}
If we subtract the constant $|\Ncal(U_i)| +
  \alpha (|U_i| +1)$ from $g_i(u)$, we obtain the vertex cost
  \begin{equation}
   \textrm{cost}_i(u) := |\Ncal(U_i\cup\{u\})|-|\Ncal(U_i)|.
  \end{equation}
  This is the number of new vertices that would be added to the neighbor
  set of partition $i$ due top adding vertex $u$ to $i$.

  When adding $u$ to a partition $i$, only the costs of a few vertices
  will be changed. Denote by
  $$\Delta_i := \cbr{v \in \Ncal(u) : v \notin \Ncal(U_i)}$$
  the set of new vertices will be added into the
  neighbor set of $U_i$ when assigning $u$ to partition $i$. Only vertices
  in $U$ connected to vertices in $\Delta_i$ will have their costs
  affected, and these costs will only be reduced and never increased. Due to the
  sparsity of the graph, this is often a small subset of the total vertices.
  Hence the overhead of updating the vertex costs is much smaller
  than re-computing them repeatedly.

{\bfseries Fast vertex cost lookup.}
We build an efficient data structure to store the vertex costs, which is
illustrated at Figure~\ref{fig:array}.  For each partition $i$, we use an array
$\Acal_i$ to store the $j$-th vertex cost, $\textrm{cost}_i(u_j)$, in the $j$-th
entry, denoted by $\Acal_i(u_j)$.  We then impose a doubly-linked list on top of
this array in an increasing  order to rapidly locate the lowest-cost vertex.
When a vertex cost is modified (always reduced), we update the doubly-linked
list to preserve the order.

Note that most large-scale graphs have a power-law degree distribution.
Therefore a large portion of vertex costs will be small integers, which are
always less equal that their degrees. We store a small array of ``head''
pointers to the locations in the list where the cost jumps to $0,1,2, \ldots
\theta$. The pointers accelerate locating elements in the list when updating. In
practice, we found $\theta = 1000$ covers over 99\% of vertex costs.

The algorithm is illustrated in Algorithm~\ref{algo:assign_a_blk}. The inputs
are a bipartite graph $G=(U,V,E)$, the number of partitions $k$, together with
$k$ sets $S_i\subseteq V$, which is the union neighbor set of vertices have been
assigned to partition $i$ before. The outputs are the $k$ partitions
$U=\bigcup_i U_i$ and updated $S_i$ with the neighbor set of $U_i$
included. Here we assume $R=U$, the sampling strategy of $R$ will be addressed
in next section.

\textbf{Runtime.} The initial $\Acal_i(u)$ can be computed in $\Ocal(|E|)$
time and then be ordered in $\Ocal(|U|)$ by counting sort, as they are integers,
upper bounded by the maximal vertex degree. The most expensive part of
Algorithm~\ref{algo:assign_a_blk} is updating $\Acal_i$ in
step~\ref{step:update}. This is are evaluated at most $k|E|$ times because, for
each partition, a vertex $v\in V$ together with its neighbors is accessed at
most once.

For most cases, the time complexity of updating the doubly-linked
lists is $\Ocal(1)$. The cost to access the $j$-th vertex is
$\Ocal(1)$ due to the sequential storing on an array.  Finding a
vertex with the minimal value or removing a vertex from the list is
also in $\Ocal(1)$ time because of the doubly links.  Keeping the list
ordered after decreasing a vertex cost by $1$ is $\Ocal(1)$ in most
cases ($\Ocal(|U|)$ for the worst case), as discussed above, by using
the cached head pointers.

The average time complexity of Algorithm~\ref{algo:assign_a_blk} is then
$\Ocal(k|E|)$, much faster than the naive implementation and orders of magnitude
better than Algorithm~\ref{algo:prototype}.

\subsection{Division into Subgraphs}
\label{sec:divide-into-subgr}

One goal of the sampling strategy used in Algorithm~\ref{algo:prototype} is to
keep the partitions of $U$ balanced, because the vertices assigned to a
partition at a time is being limited. The additional constraint $|T|=1$
introduced in the previous section ensures that only a single vertex is assigned each
time, which addresses balancing. Consequently we would like
to sample as many vertices as possible to enlarge the search range of the
optimal $u^*$ to partition quality.
Sampling remains appealing since it is a trade-off between computation
efficiency and partition quality.

\algo first randomly divides $U$ into $b$ blocks. It next constructs the
corresponding $b$ subgraphs by adding the neighbor vertices from $V$ and the
corresponding edges, and then partitions these subgraphs sequentially by
Algorithm~\ref{algo:assign_a_blk}. In other words, denote by $\{G_j\}_{j=1}^b$
the $b$ subgraphs, and $\{S_i\}_{i=1}^k$ the initialized neighbor sets, for
instances $S_i = \emptyset$ for all $i$; at iteration $j=1,\ldots,b$, we
sequentially feed $G_j$ and $S_i$'s into Algorithm~\ref{algo:assign_a_blk} to
obtain the partitions $\bigcup_{i=1}^k U_{i,j}$ of $G_j$ and updated $S_i$'s,
which contain the previous partition information.  Then we union the results on
each subgraph to the final partitions of $U$ by $U_i =\bigcup_{j=1}^b U_{i,j}$
for $i=1,\ldots,k$.

Compared to the scheme described in Section~\ref{sec:partition-U} which samples
a new subgraph ($R$) for each single vertex assignment, \algo fixes those
subgraphs at the beginning.  This sampling strategy has several
advantages. First, it partitions a subgraph by directly using
Algorithm~\ref{algo:assign_a_blk}, which takes advantage of the head pointers
and linked list to improve the efficiency. Next, it is convenient to place both
the initialization of neighbor sets and parallelization which will be introduced
soon on subgraph granularity. Finally, this strategy is I/O efficient, because
we must only keep the current subgraph in memory. As a result, it is possible to
partition graphs of sizes much larger than physical memory.

The number of subgraphs $b$ is a trade-off between partition quality and
computational efficiency. In the extreme case of $b=1$, the vertex assigned to
a partition is the optimal one from all unassigned vertices. It is, however, the
most time consuming. In contrast, though the time complexity reduces to
$\Ocal(|E|)$ when letting $b=|U|$, we only get random partition
results. Therefore, a well-chosen size $b$ not only removes the graph size
constraint but also balances time and quality.

\subsection{Parallelize  with the Parameter Server}

Although Parsa can partition very large graphs with a single process by taking
advantage of sampling, parallelization is desirable because of the reduction of
both CPU and I/O times on each machine. Parsa parallelizes the
partitioning by processing different
subgraphs in parallel (on different nodes) by using the shared neighbor sets.

To implement the algorithm using the \ps, we need the following three
groups of nodes:
\begin{description}
\item[The scheduler] issues partitioning tasks to workers and monitors their progress.
\item[Server nodes] maintain the global shared neighbor sets. They process
   push and pull requests by workers.
\item[Worker nodes] partition subgraphs in parallel. Every time a worker
  first reads a subgraph from the (distributed) file system. It then pulls the
  newest neighbor sets associated with this subgraph from the servers. Then, it
  partitions this subgraph using Algorithm~\ref{algo:assign_a_blk} and finally
  pushes the modified neighbor sets to the servers.
\end{description}

\subsection{Initializing the Neighbor Sets}

The neighbor sets play a similar role as cluster centers on clustering methods,
both of which affect the assignment of vertices. Well-initialized neighbor sets
potentially improve the partition results. Initialization
by empty sets, which prefers assigning vertices with small degrees
first, however, often helps little, or even
degrades, the resulting assignment.
Parsa uses several initialization
strategies to improve the results:
%
\begin{description}
\item[Individual initialization.] Given a graph that has been divided into $b$
  subgraphs, we can runs $a+b$ iterations where the results for the first $a$
  iterations are used for initialization. In other words, before processing the
  $(j+1)$-th subgraph, $j\le a+1$, we reset the neighbor set by $S_i =
  \Ncal(U_{i,j})$, where $\{U_{i,j}\}_{i=1}^k$ are the partitions of $j$-th
  subgraph.  The old results are dropped because otherwise a vertex $u$ will be
  assigned to its old partition $i$ again as $S_i$ contains the neighbors of $u$
  and the cost $|\Ncal(u)\setminus S_i|$ will then be 0.
\item [Global initialization.] In parallel partitioning, before
  starting all workers, we first sample a small part from the graph
  and then let one worker partition this small subgraph. Then we can
  use the resulting neighbor sets as an initialization to all workers.
\item [Incremental partitioning.] In this setting, data arrives in an
  incremental way and we want to partition the new data efficiently. Since we
  already have the partitioning results on the old data, we can use
  these results as initialization of the neighbor sets.
\end{description}

\subsection{Puting it all together}

Algorithm~\ref{algo:parallel} shows Parsa, which partitions $U$ into $k$ parts
in parallel. Then we can assign $V$ using Algorithm~\ref{algo:part_V} if
necessary.  The initial neighbor sets can be obtained from global initialization
or incremental partitioning discussed in the previous section.  There
are several details worth noting:
First, while communication in the \ps is asynchronous, Parsa imposes a
maximal allowed delay $\tau$ to control the data consistency.
%
Second, the worker might only push the changes of the neighbor sets to the servers to save the communication traffic.
Finally, a worker may start a separate data pre-fetching thread to run
steps 3, 4 and 5 to improve the efficiency.

\begin{algorithm}[t!]
  \caption{\algo: parallel submodular approximation}
  \label{algo:parallel}
  \begin{algorithmic}[1]
    \REQUIRE Graph $G$, initial neighbor sets $\cbr{S_i}_{i=1}^k$, \#partitions $k$, max delay $\tau$, initialization from
    $a$, \#subgraphs $b$.
    \ENSURE partitions $U=\bigcup_{i=1}^k U_i$
  \end{algorithmic}
  {\bf Scheduler:}
  \begin{algorithmic}[1]
    \STATE divide $G$ into $b$ subgraphs
    \STATE ask all workers to partition with
    $(a, \tau, \textrm{true})$
    \STATE ask all workers to partition with
    $(b, \tau, \textrm{false})$
  \end{algorithmic}
  {\bf Server:}
  \begin{algorithmic}[1]
    \STATE start with a part of $\cbr{S_i}_{i=1}^k$
    \IF {receiving a { pull} request}
    \STATE reply with the requested neighbor set $\cbr{S_i}_{i=1}^k$
    \ENDIF
    \IF {receiving a { push} request containing $\cbr{S_i^{\textrm{new}}}_{i=1}^k$}
    \IF {$\textrm{initializing}$}
    \STATE $S_i \gets S_i^{\textrm{new} }$ for $i = 1, \ldots, k$
    \ELSE
    \STATE $S_i \gets S_i \cup S_i^{\textrm{new}}$ for $i = 1, \ldots, k$
    \ENDIF
    \ENDIF
  \end{algorithmic}
  {\bf Worker:}
  \begin{algorithmic}[1]
    \STATE receive hyper-parameters $(T, \tau, \textrm{initializing})$
    \FOR {$t = 1, \ldots, T$}
    \STATE load a subgraph $G(U,V,E)$
    \STATE wait until all {push}es before time $t-\tau$ finished
    \STATE {pull} the part of neighbor sets, $\cbr{S_i}_{i=1}^k$, that
    contained in $V$ from the servers
    \STATE get partitions $\cbr{U_i^{\textrm{new}}}_{i=1}^k$ and updated neighbor sets
    $\cbr{S^{\textrm{new}}_i}_{i=1}^k$ using Algorithm~\ref{algo:assign_a_blk}
    \IF {$\textrm{initializing} = \textrm{false}$ and $t>1$ }
    \STATE $S^{\textrm{new} }_i \gets S^{\textrm{new} }_i \;\backslash\; S_i$ for all $i=1,\ldots, k$
    \ENDIF
    \STATE {push} $\cbr{S^{\textrm{new} }_i}_{i=1}^k$ to servers
    \STATE {\bf if not} initializing  {\bf then } $U_i \gets U_i \cup U_i^{\textrm{new} }$
    for all $i$ 
    \ENDFOR
  \end{algorithmic}
\end{algorithm}

\section{Experiments}
\label{sec:experiments}

We chose 7 datasets of varying type and scale, as summarized in
Table~\ref{tab:dataset}. The first three are text
datasets\footnote{\small\url{http://www.csie.ntu.edu.tw/~cjlin/libsvmtools/datasets/}};
live-journal and orkut are social
networks\footnote{\small\url{http://snap.stanford.edu/data/}}; and the last two
are click-through rate datasets from a large Internet company. The numbers of
vertices and edges range from $10^4$ to $10^{10}$.

\subsection{Setup}

\def\ctra{\textsf{CTRa}\xspace}

\def\lj{\textsf{live-journal}\xspace}
\def\orkut{\textsf{orkut}\xspace}
\def\news{\textsf{news20}\xspace}
\def\kdda{\textsf{KDDa}\xspace}
\def\rcv{\textsf{rcv1}\xspace}

\begin{table}[t]
  \centering
  \begin{tabular}{|l|rrrr|}
    \hline
    name         & $|U|$ & $|V|$ & $|E|$ & type       \\
    \hline
    \textsf{rcv1} & 20K  & 47K  & 1M   & bipartite  \\ 
    \news         & 20K  & 1M   & 9M   & bipartite  \\ 
    \kdda         & 8M   & 20M  & 305M & bipartite  \\ 
    \lj           & 5M   & 5M   & 69M  & directed   \\
    \orkut        & 3M   & 3M   & 113M & undirected \\
    \ctra         & 1M   & 4M   & 120M & bipartite  \\
    \ctrb         & 100M & 3B   & 10B  & bipartite  \\
    \hline
  \end{tabular}
  \caption{A collection of real datasets. }
  \vspace{-2mm}
  \label{tab:dataset}
\end{table}

{
  \def\mem{${M}_{\textrm{max}}$\xspace}
  \def\tmax{${T}_{\textrm{max}}$\xspace}
  \def\tsum{${T}_{\textrm{sum}}$\xspace}
  \def\headb{\multicolumn{3}{c|}{imprv (\%) } & {\small time} }
  \def\headc{\mem & \tmax & \tsum & {\small (sec)}}
\begin{table*}[t!]
   \setlength{\tabcolsep}{1.5pt}
  \centering
  \begin{tabular}{|l|rrr|r|rrr|r|rrr|r|rrr|r|rrr|r|}
    \hline
    \multirow{3}{*}{datasets} & \multicolumn{4}{c |}{PowerGraph}  & \multicolumn{4}{c |}{METIS}  &\multicolumn{4}{c |}{PaToH} &
    \multicolumn{4}{c |}{Zoltan} & \multicolumn{4}{c |}{\algo }    \\
    \cline{2-21}
    & \headb & \headb & \headb & \headb & \headb \\
    & \headc & \headc  & \headc  & \headc & \headc \\
    \hline
    \rcv   & -  & -  & -  & -       & -        & -        & -        & -   & \sd{19} & 26       & \ft{279}  & 7         & 17      & \sd{105}  & \sd{154}  & \sd{6}   & \ft{33}  & \ft{112}  & 108       & \ft{0.2} \\
    \news  & -  & -  & -  & -       & -        & -        & -        & -   & \sd{2}  & 123      & \ft{389}  & 23        & 0       & \ft{214}  & \sd{267}  & \sd{21}  & \ft{23}  & \sd{187}  & 155       & \ft{1}   \\
    \ctra  & -  & -  & -  & -       & -        & -        & -        & -   & 8       & 446      & \sd{970}  & 571       & \sd{70} & \ft{1052} & \ft{1211} & \sd{551} & \ft{91}  & \sd{922}  & 913       & \ft{18}  \\
    \kdda  & -  & -  & -  & -       & -        & -        & -        & -   & \sd{54} & \sd{905} & \sd{1102} & \sd{1401} & 4       & 238       & 313       & 2409     & \ft{120} & \ft{1973} & \ft{1978} & \ft{89}  \\
    \lj    & 61 & 84 & 89 & \ft{9}  & \ft{185} & \ft{231} & \sd{279} & 65  & 103     & 152      & 160       & 3.5h      & 50      & 84        & \ft{386}  & 1072     & \sd{142} & \sd{216}  & 214       & \sd{37}  \\
    \orkut & 55 & 74 & 78 & \ft{12} & 56       & 74       & 103      & 104 & \sd{87} & 145      & 150       & 5.5h      & 49      & \sd{170}  & \ft{180}  & 1413     & \ft{105} & \ft{177}  & \sd{121}  & \sd{39}  \\
    \hline
  \end{tabular}
  \caption{Improvements (\%) comparing to random partition on the maximal individual
    memory footprint \mem, maximal individual
    traffic volumes \tmax, and total traffic volumes \tsum together with running times (in sec) on
    16-partition. The best results are colored by \ft{Red} and the second best by  \sd{Green}.}
  \label{tab:st}
\end{table*}
}
\begin{figure}[tb]
  \noindent%
  \includegraphics[width=\columnwidth]{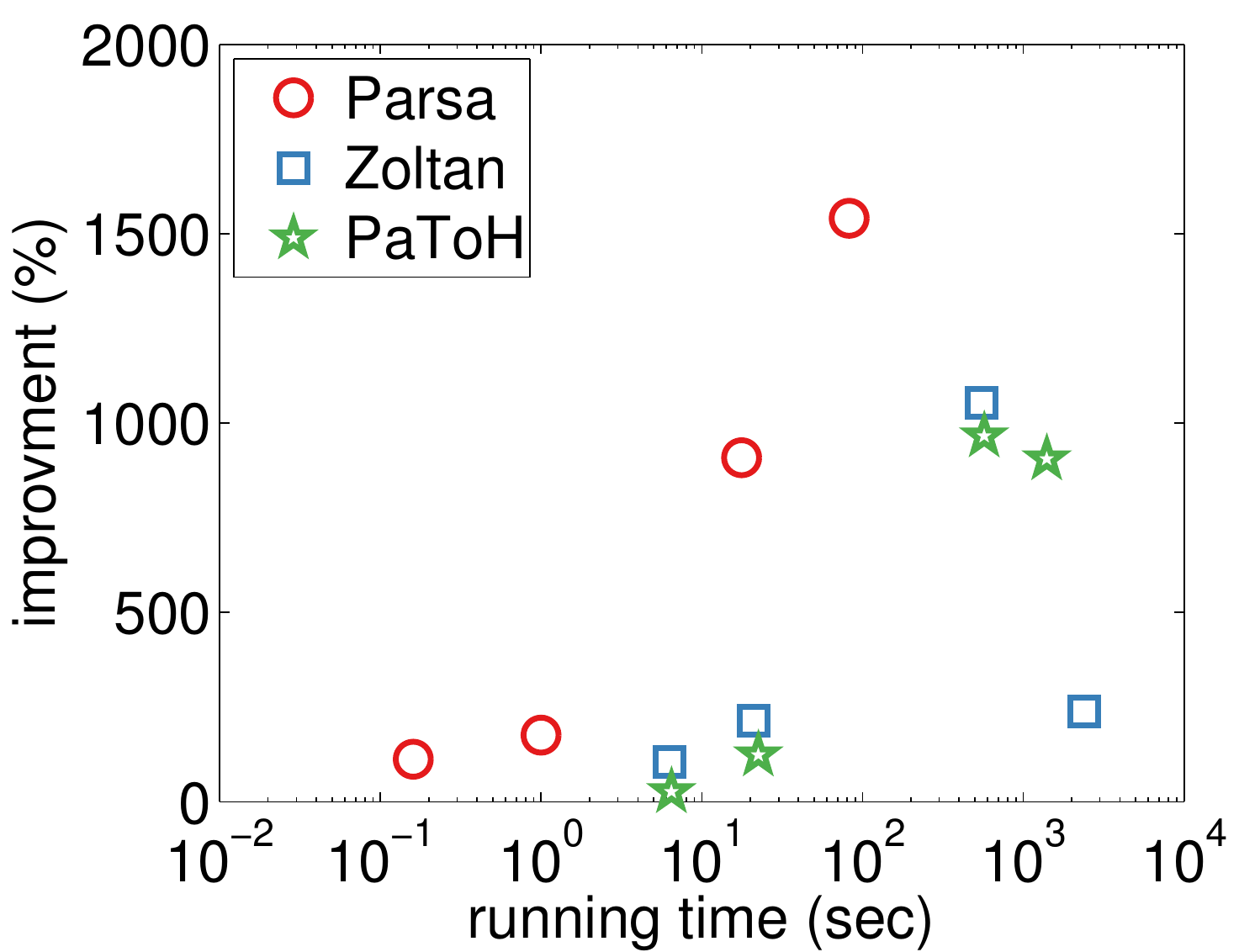}
  \includegraphics[width=\columnwidth]{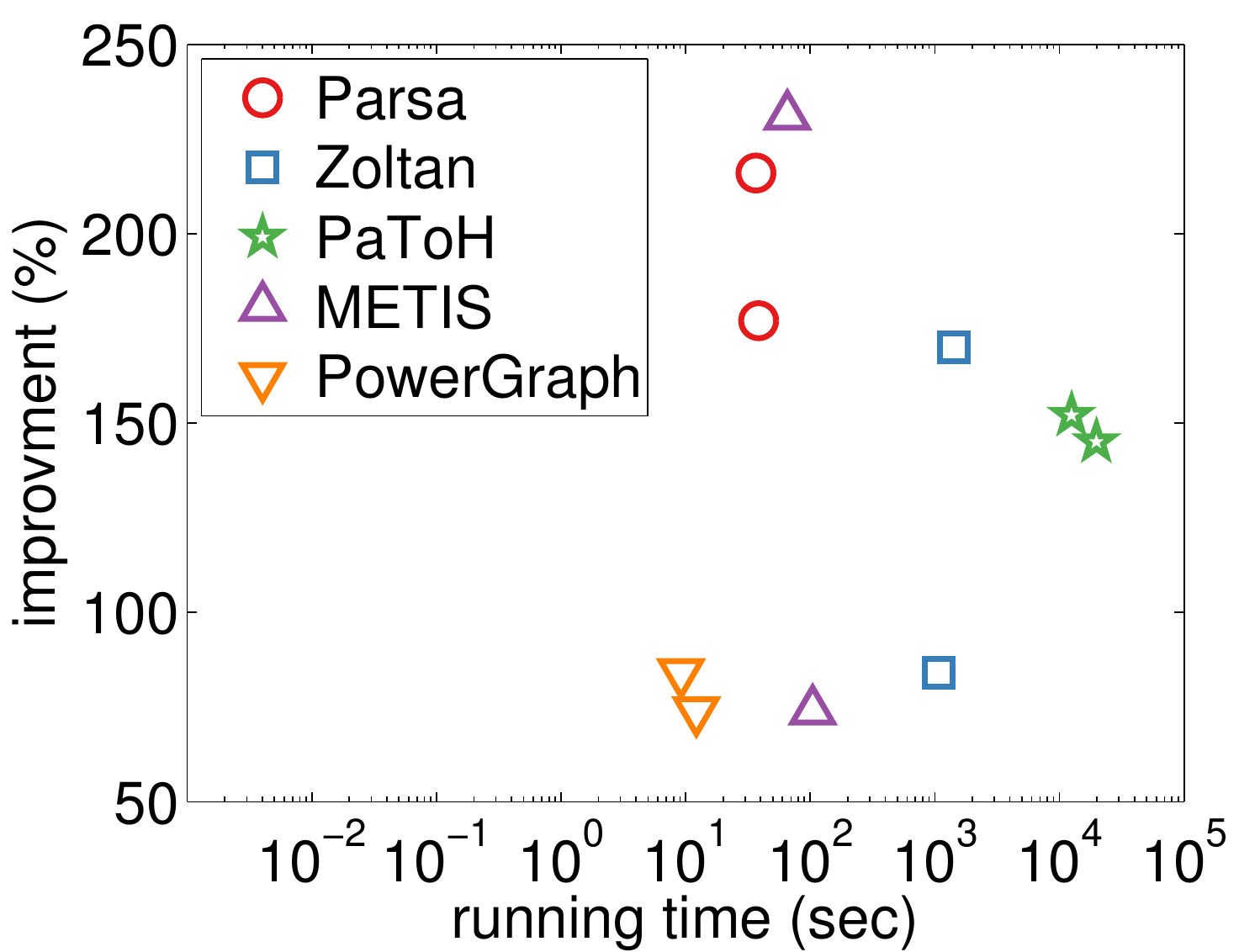}
  \caption{Visualization of comparisons from Table~\ref{tab:st}. Top:
    text datasets; Bottom: Social Networks; Note that on text datasets
    (naturally bipartite graphs) Parsa is both orders of magnitude
    faster and yields better results. On social networks PowerGraph is
  faster, which is to be expected since it uses very simple and fast
  partitioning.}
  \label{fig:cmp_res}
\end{figure}

We implemented \algo in the \ps~\cite{LiAndParSmoetal14}; the source
is available at \url{https://github.com/mli/parsa}. We compared \algo with
popular vertex-cut graph partition toolboxes
Zoltan\footnote{\small\url{http://www.cs.sandia.gov/Zoltan/}} and
PaToH\footnote{\small\url{http://bmi.osu.edu/~umit/software.html}}, which can
also take bipartite graphs as inputs. We also used the well-known graph
partition package
METIS\footnote{\small\url{http://glaros.dtc.umn.edu/gkhome/metis/metis/overview}}
and the greedy algorithm adopted by
Powergraph\footnote{\small\url{http://graphlab.org/downloads/}}, though these
handle only normal graphs. All algorithms are implemented in C/C++.

We report both runtime and partition results.  The default
measurement of the latter is the maximal individual traffic volume. We counted
the improvement against random partition by
${(\text{random}-\text{proposed})}/{\text{proposed}} \times 100\%,$
where 100\% improvement means that traffic or memory footprint are
50\% of that achieved by random partitioning.

The default number of partitions was set to $16$. As \algo is a randomized
algorithm, we recorded the average results over 10 trials. Single thread
experiments used a desktop with an Intel i7 3.4GHz CPU, while the parallel
experiments used a university cluster with 16 machines, each with an Intel Xeon
2.4GHz CPU and 1 Gigabit Ethernet.

\subsection{Comparison to other Methods}

Table \ref{tab:st} shows the comparison results on different datasets.
We recorded the CPU time on running each algorithm except for loading the data,
because its performance varies for different data formats each algorithm used.
The improvements are measured on maximal individual memory footprint and traffic
volume, together with total traffic volume, which is the objective for both
PaToH and Zoltan.  Since neither METIS nor PowerGraph handle general sparse
matrices, only results on social networks are reported. The number of partitions
is 16, and the parameters of \algo are fixed by $a=b=16$. See
Figure~\ref{fig:cmp_res} for improvements on maximal individual
traffic volume and runtimes.

As can be seen, \algo is not only 20x faster than PaToH and Zoltan,
but also produces more stable partition results, especially on
reducing the memory footprint. METIS outperforms \algo on one of
the two social networks but consumes twice as much CPU
time. PowerGraph is the fastest but suffers the cost of worse
partition quality. Under both measurements
on maximal individual traffic volume and total traffic volume, \algo
produces similar results.

As the number of partitions increases, the recursive-bisection-based
algorithms (METIS, PaToH, and Zoltan) retain their runtimes, but their
partition quality degrades, as shown in Figure~\ref{fig:vary_k}.  In
contrast, \algo and Powergraph compute $k$-partitions directly.  Their
runtimes increase linearly with $k$, but their partition quality
actually improves.

\begin{figure*}[tbh!]
  \includegraphics[height=0.33\textwidth]{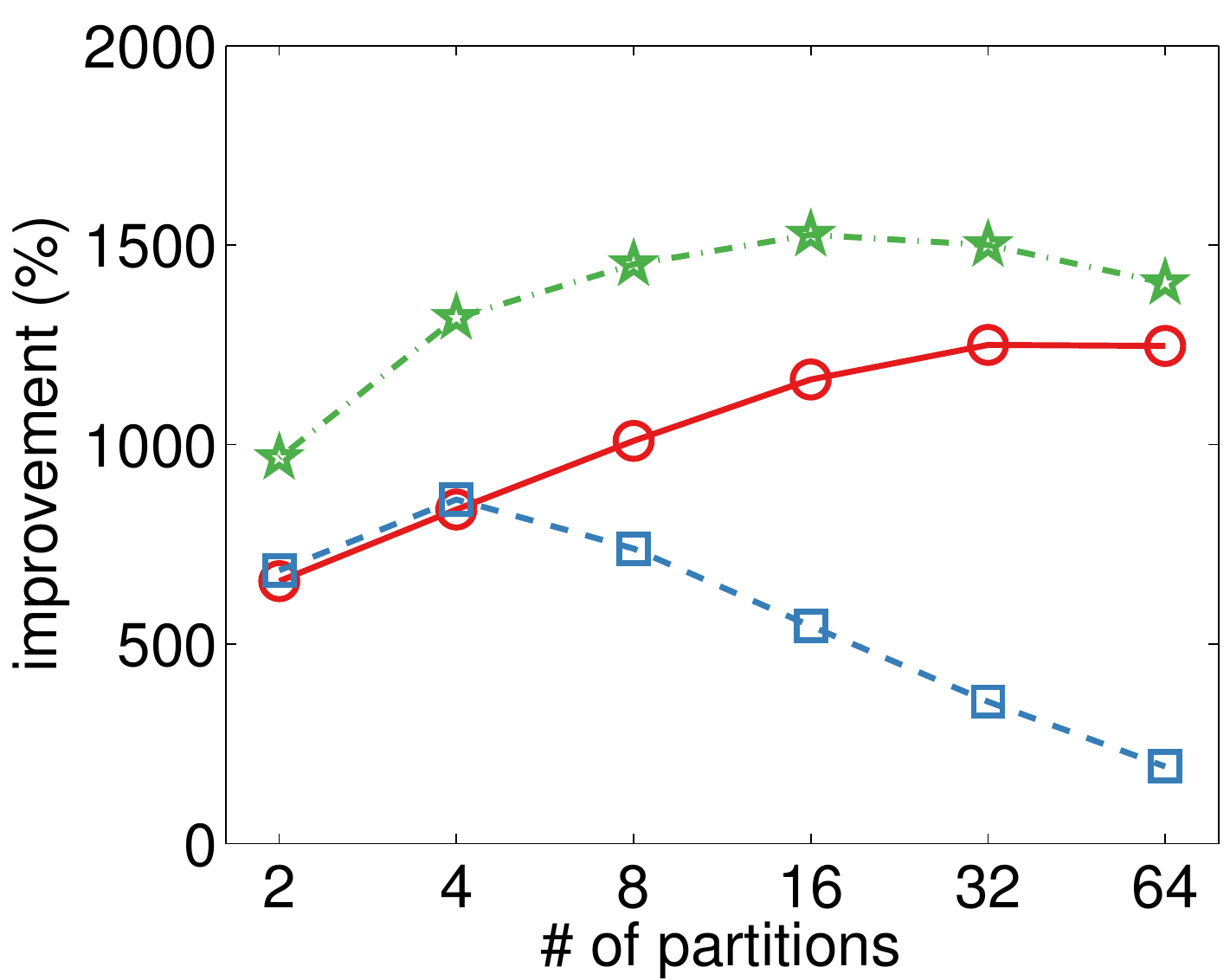} 
  \includegraphics[height=0.33\textwidth]{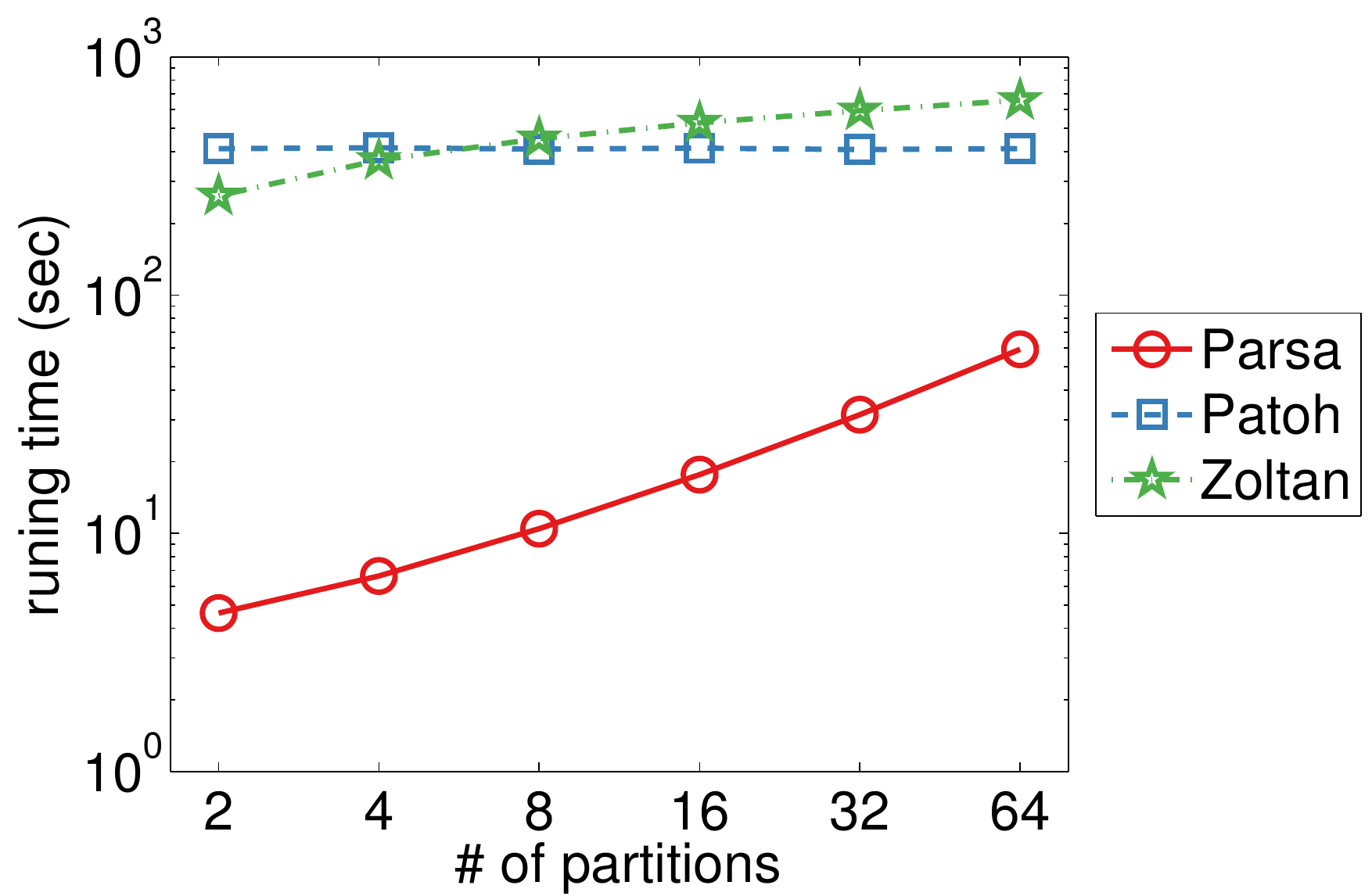} \\
  \includegraphics[height=0.33\textwidth]{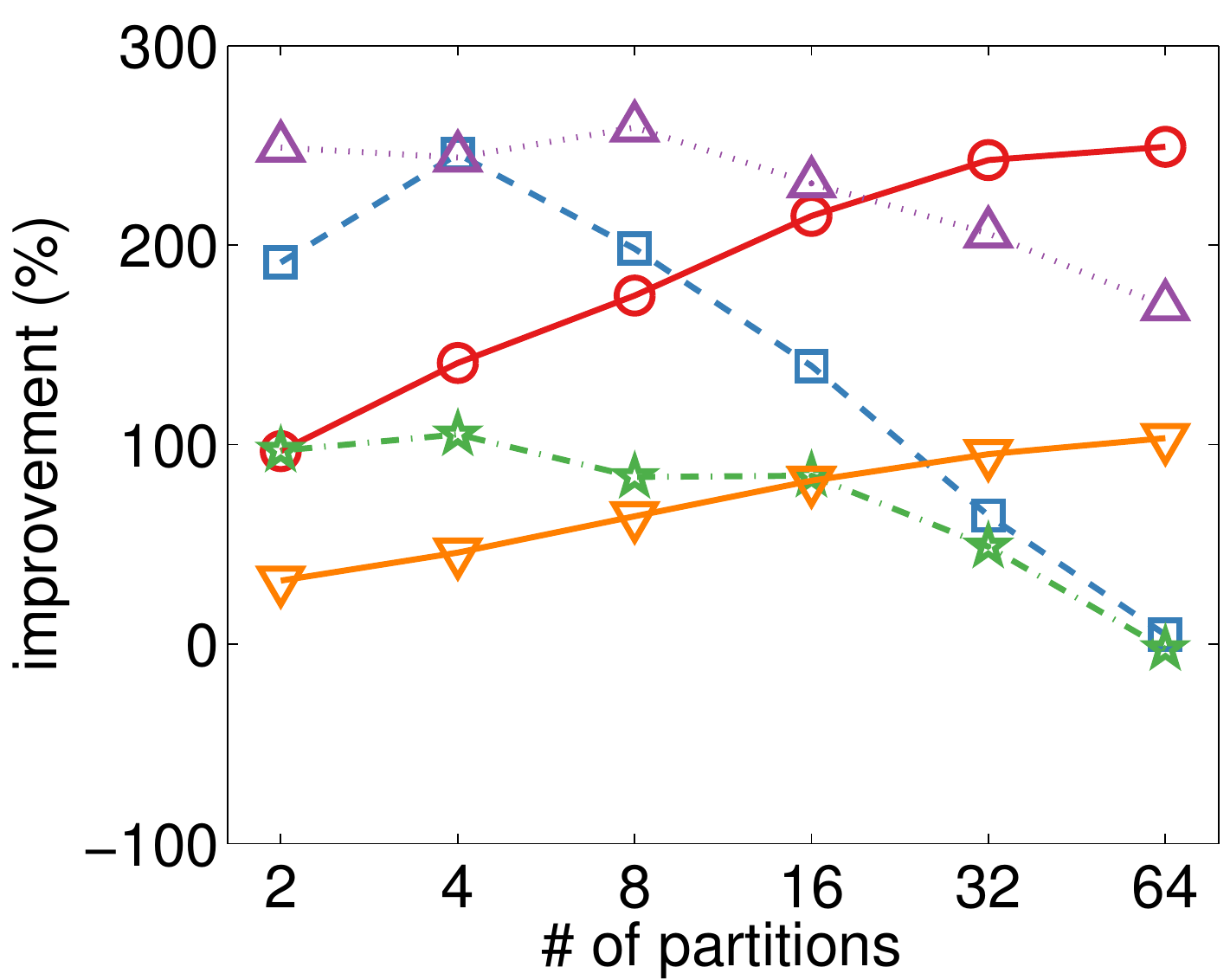}
  \includegraphics[height=0.33\textwidth]{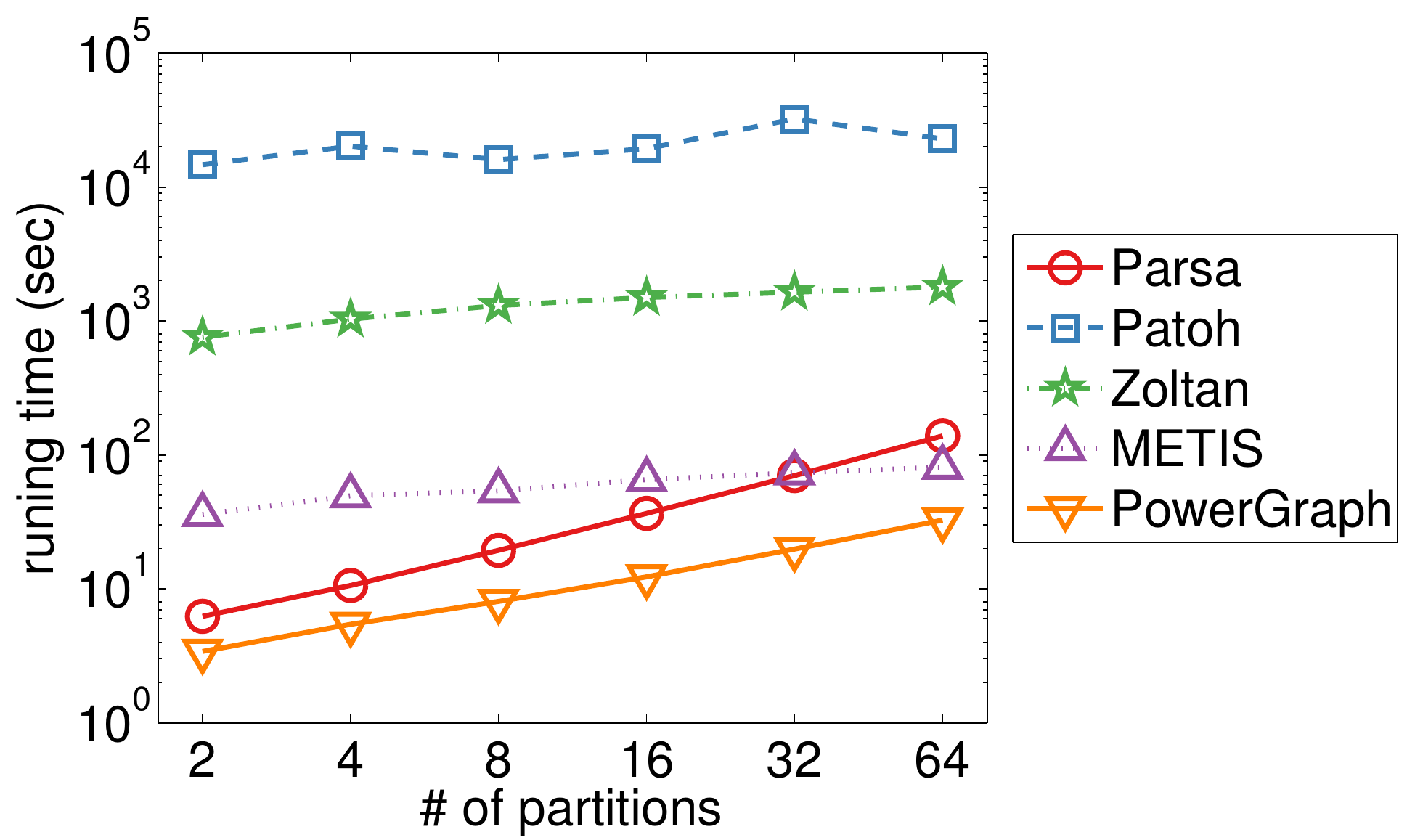}
  \caption{Improvement over random partitioning when changing the
    number of partitions $k$. Top: \ctra; Bottom: \lj. Note that even both Parsa
    and PowerGraph use more time when $k$ increases, they
    improve the partition quality.}
  \label{fig:vary_k}
\end{figure*}

\subsection{Number of subgraphs and initialization}

\begin{figure*}[tb]
  \noindent%
    \includegraphics[width=0.49\textwidth]{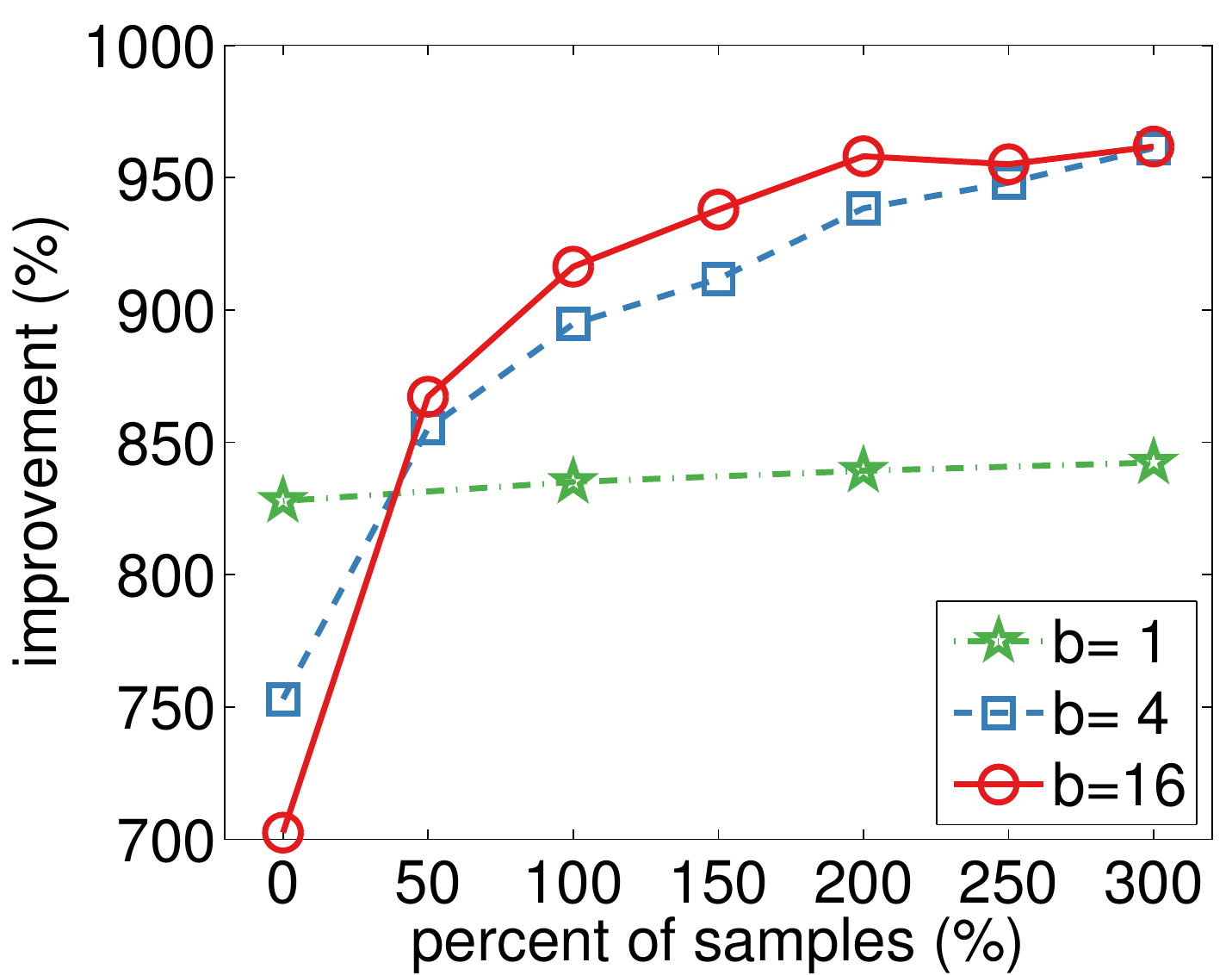}
    \includegraphics[width=0.49\textwidth] {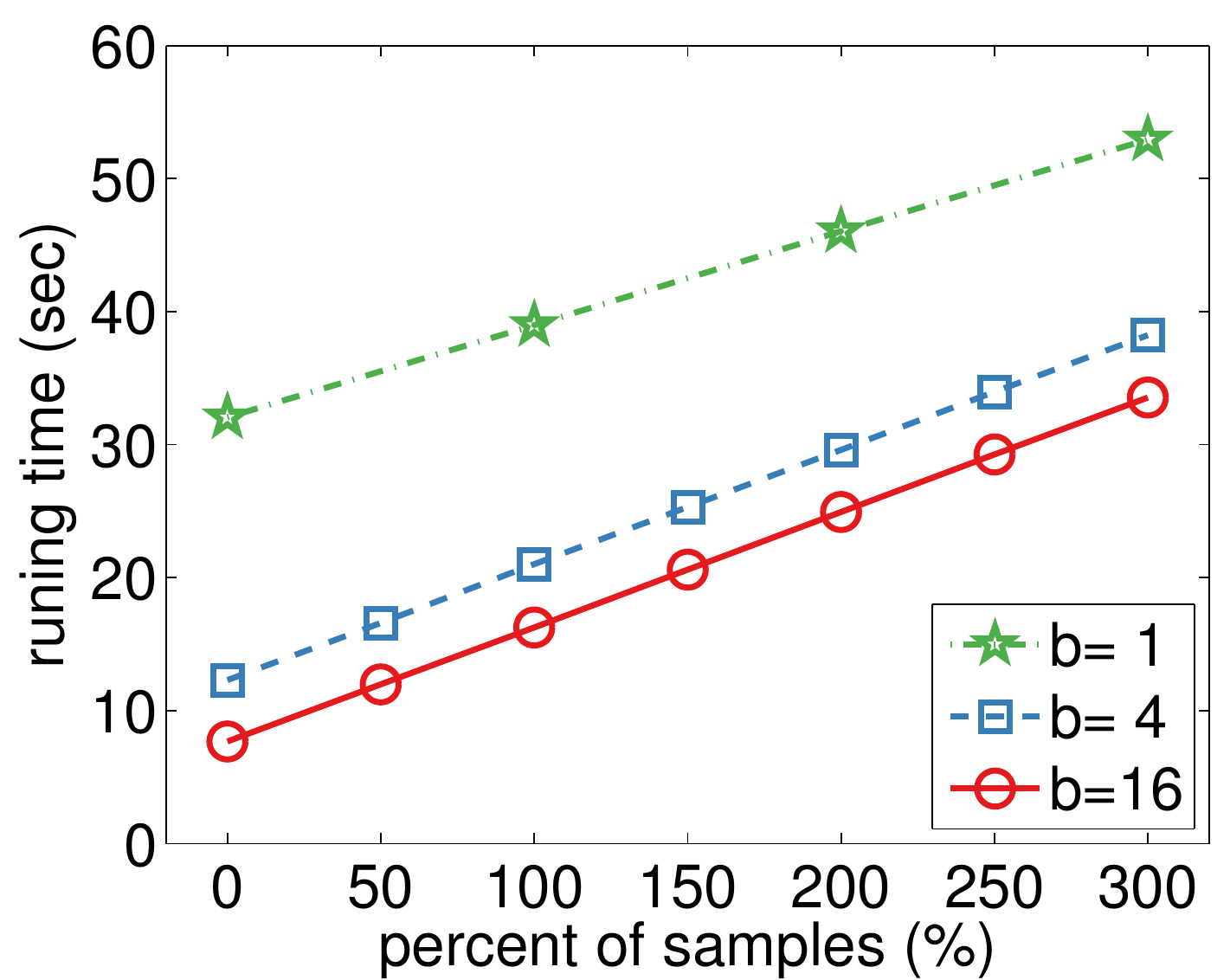} \\
    \includegraphics[width=0.49\textwidth] {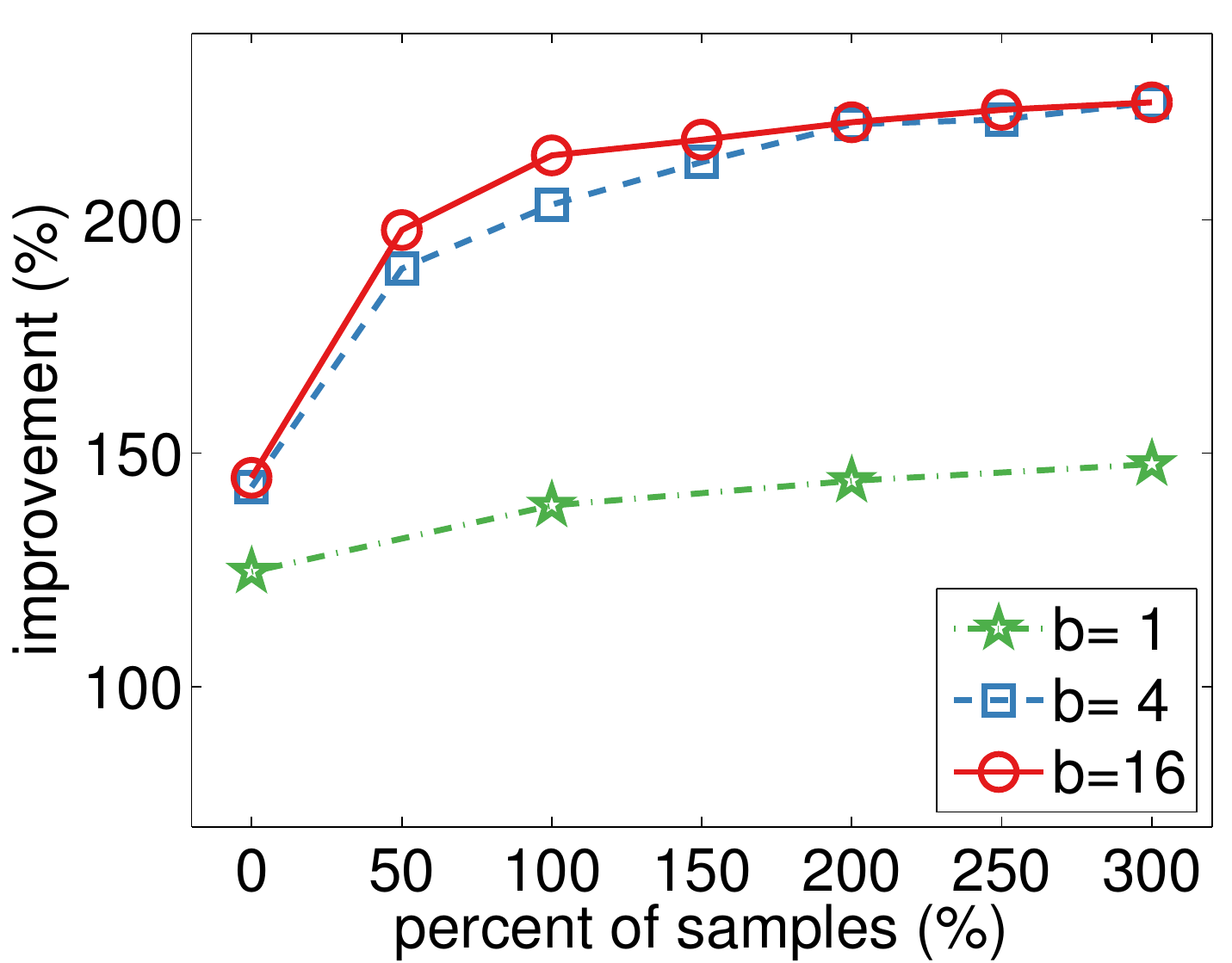}
    \includegraphics[width=0.49\textwidth] {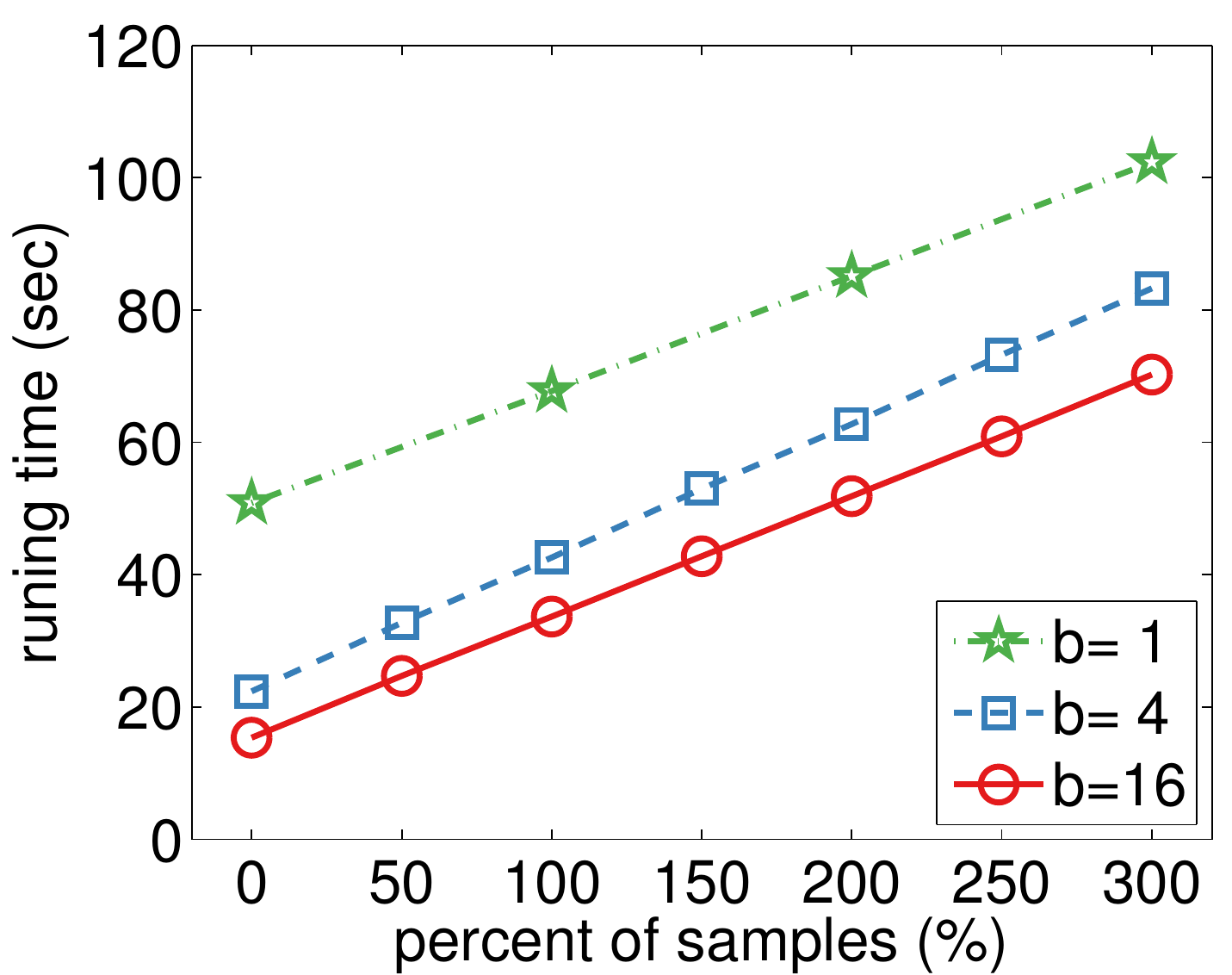}
  \caption{Varying the number of subgraphs and percent of data used for
    initialization for single thread partitioning. Top: \ctra; Bottom:
    \lj.
\label{fig:init}}
\end{figure*}

We examine two important optimizations in Parsa: the size of
subgraph and the neighbor set initialization. We first consider the single thread
case, which  starts with empty neighbor sets. We divide the data into different
numbers of subgraphs $b$ and also varying the number of subgraphs $a$ used for
initialization.  The results on representative text dataset CTRa and social
network live-journal are shown in Figure~\ref{fig:init}.

The x-axis plots
$a/b\times 100\%$, which is the percent of data used for constructing the
initialization. It is clear that using
more data for initialization improves partition quality.  Although the
improvement is not significant for the single subgraph case ($b=1$) because
partitioning the same subgraph several times changes the results little, there
is a stable 20\% improvement over no initialization when at least 100\% data are
used for $b > 1$.

Without initialization, using small subgraphs has a positive effect on partition
results for live-journal, but not for CTRa. The reason, as mentioned in
Section~\ref{sec:divide-into-subgr}, is that \algo prefers to assign vertices
with small degrees first when starting with empty neighbor sets.  Those vertices
offer little or no benefit for the subsequent assignment.  Live-journal has many
more sparse vertices than CTRa due to the power law distribution, and
partitioning small subgraphs reduces the number of sparse vertices entering
partitions too early.

Initialization solves the previous problem by dropping early partition results
and resetting corresponding neighbor sets.  With $a >
\frac{1}{2} b$ in Figure~\ref{fig:init}, small subgraphs improve the partition
results on both CTRa and live-journal.  This occurs since when using the same
percentage of data for initialization, neighbor sets with small $b$
are reset more often.

Figure~\ref{fig:init} also shows the runtime.  Splitting into more blocks
(larger $b$) narrows the search range for adding vertices, which reduces the
cost of operating the doubly-linked list, boosting speed.  The runtime increases
linearly as we use and discard more samples for initialization, but the
partition quality benefits of doing so appear worthwhile up to performing two
passes (100\% samples).

Next we consider the parallel case with non-empty starting neighbor sets. We use
4 workers to partition a subset of \ctrb containing 1 billion of edges and use
one worker to partition an even smaller subgraph to obtain the
starting neighbor sets. The results of varying the size of the this subgraph are shown in
Figure~\ref{fig:vary_par_init}.

As can be seen, the partition quality is
significantly improved even when only 0.1\% data are used for the global
initialization. In addition, although this initialization takes extra time, the
total running time is minimized when we used initialization. A good
initialization of the neighbor sets reduces the cost of operating the doubly
linked lists, saving time.

\begin{figure}[tbh]
  \includegraphics[width=\columnwidth]{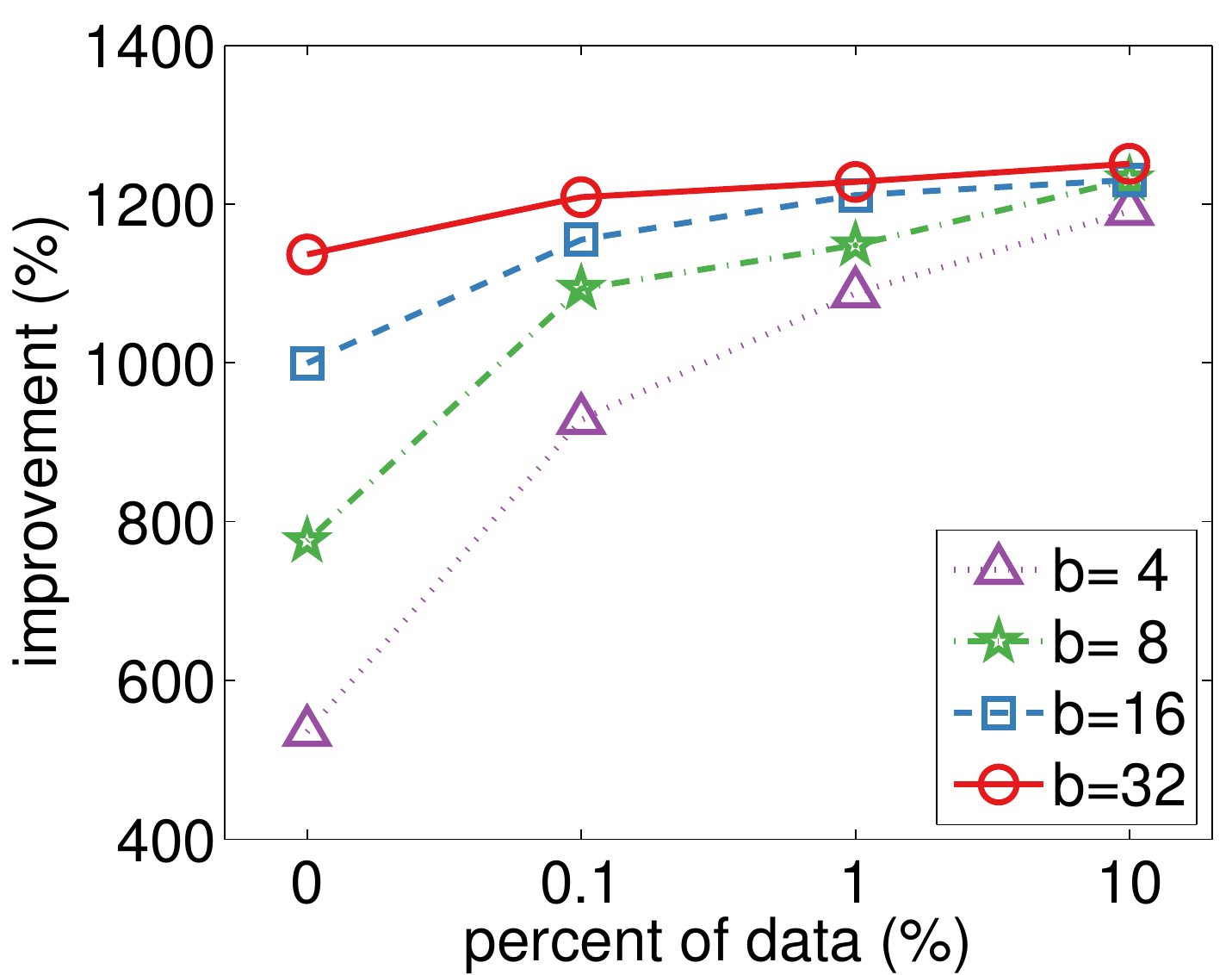}
  \includegraphics[width=\columnwidth]{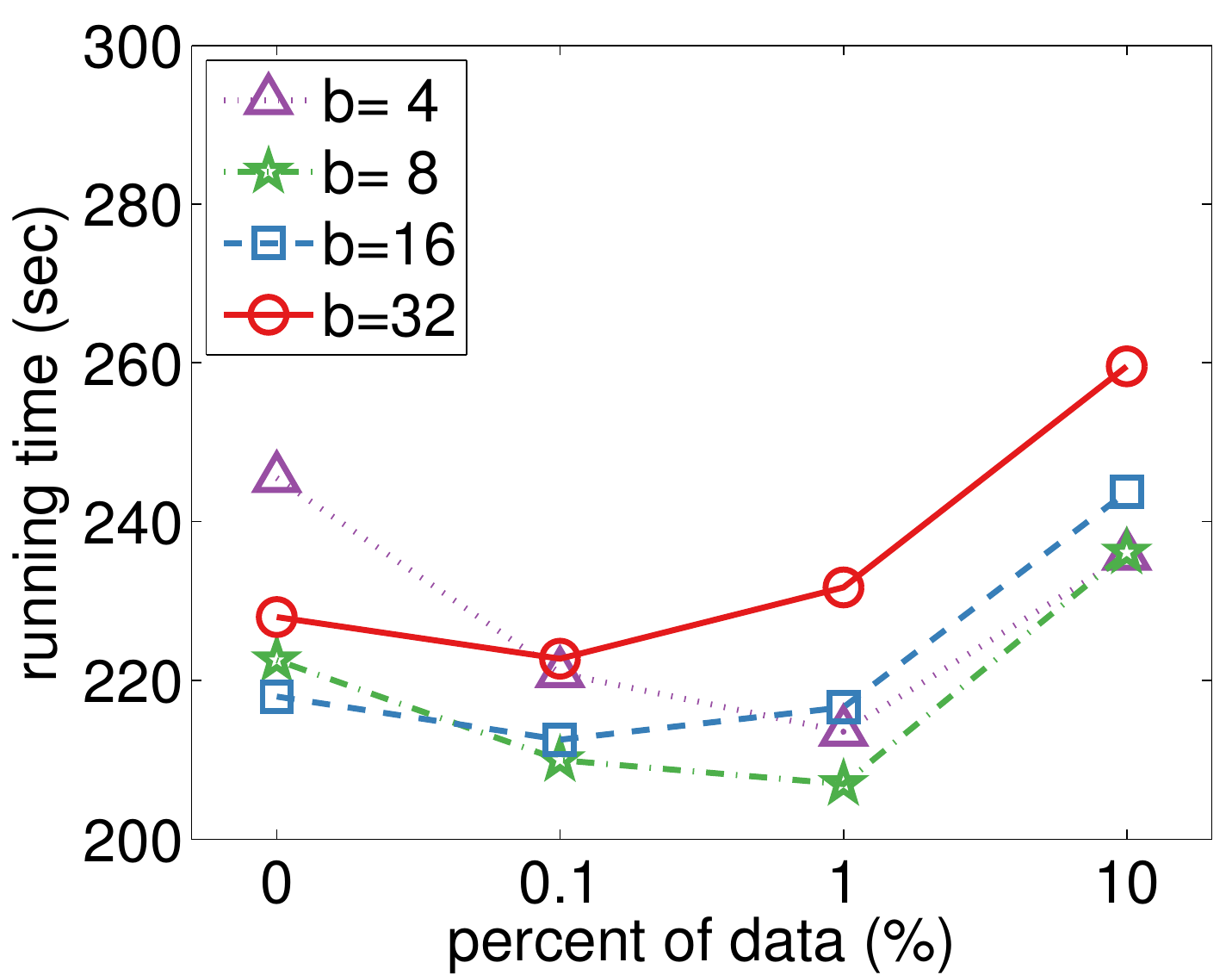}
  \caption{Varying the percentage of data used for global initializing the
    neighborhood sets with 4 workers.}
  \label{fig:vary_par_init}
\end{figure}

\subsection{Scalability}

\begin{figure}[t!]
  \centering
  \noindent%
  \includegraphics[width=\columnwidth]{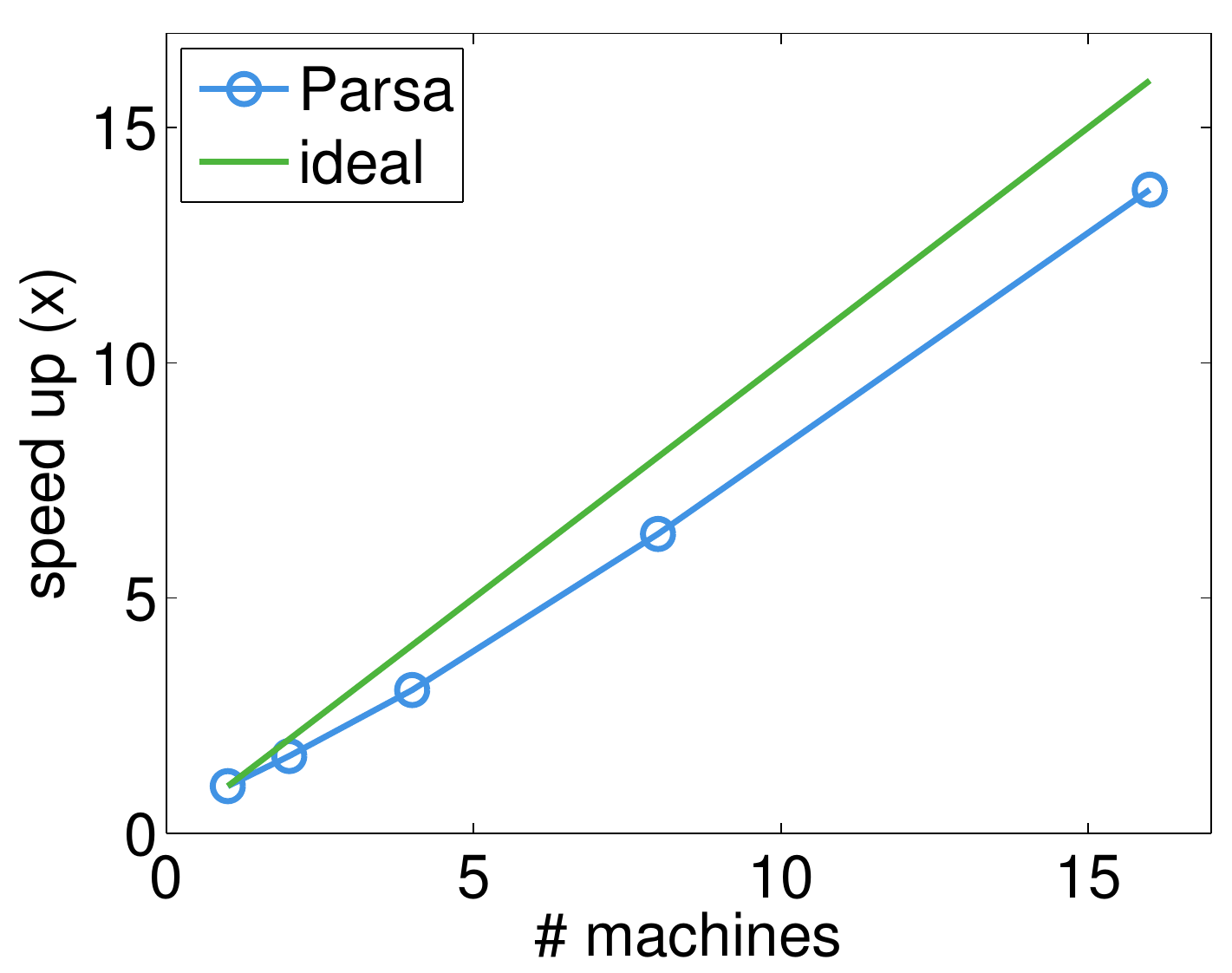}
  \caption{The scalability of Parsa on dataset \ctrb.  }
    \label{fig:scal}
\end{figure}

We test the scalability of Parsa on \ctrb with 10 billion edges by
increasing the number of machines. We run 4 workers and 4 servers at
each machine with infinite maximum delay. The results are shown in
Figure~\ref{fig:traffic_vs_datasize}. As can be seen, the speedup is
linear with the number of machines and close to the ideal case. In
particular, we obtained a 13.7x speedup by increasing the number of
machines from 1 to 16.

The main reason that Parsa scales well is due to the eventual consistency
model ($\tau=\infty$). In this model, there is no global barrier between
workers, and each worker even does not wait the previous results pushed
successfully. Therefore, workers fully utilize the
computational resource and network bandwidth, and waste no time on waiting
the data synchronization.

This consistency model, however, potentially leads to inconsistency of the
neighbor sets between workers. However, we found that Parsa is robust to this
kind of inconsistency. In our experiment, increasing the number of machines from
1 to 16 (4 to 64 in terms of workers) only decreases the quality of the
partition result at most by $5\%$. We believe the reason is twofold. First, the
starting neighbor sets obtained on a small subgraph let all workers have a
consistent initialization, which may contain the membership of most head (large
degree) vertices in $V$. Second, the modifications of the neighbor sets each
worker contributed after partitioning a subgraph therefore are mainly about tail
(small degree) vertices in $V$. Due to the extreme sparsity of the tail
vertices, the conflicts among workers could be small and therefore affect the
results little.

\subsection{Accelerating Distributed Inference}

Finally we examine how much Parsa can accelerate distributed machine learning
applications by better data and parameter placement. We consider
$\ell_1$-regularized logistic regression, which is one of the most widely
used machine learning algorithm for large scale text datasets.
We choose a
state-of-the-art distributed inference algorithm,
DBPG~\cite{LiAndSmoYu14}, to solve this application.  It is based on the block
proximal gradient method using several techniques to improve
efficiency: it supports a maximal $\tau$-delay consistency model similar to
Parsa, and uses several user-defined filters, such as key caching, value
compression, and an algorithm-specific KKT filter, to further reduce
communication cost.  This algorithm has been implemented in the
\ps~\cite{LiAndParSmoetal14}, and is well optimized. It can use
1,000 machines to train $\ell_1$-regularized logistic regression on
500 terabytes data within a hour~\cite{LiAndSmoYu14}.

We run DBPG on \ctrb using 16 machines as the baseline. We enabled all
optimization options described in~\cite{LiAndSmoYu14}.
Then we partition \ctrb
into 16 parts by Parsa and run DBPG again. The runtime is shown in
Table~\ref{tab:batch}. By random partitioning, DBPG stops after passing the
data 45 times and uses 1.43 hours. On the other hand, Parsa uses 4 minutes to
partition the data and then accelerates DBPG to 0.84 hour. As a
result, Parsa can reduce the total time from 1.43 hours to 0.91 hourx,
a 1.6x speedup.

The reason Parsa accelerates DBPG is shown clearly in
Table~\ref{tab:batch}. By random partition, only 6\% of network traffic between
servers and workers happens locally. Even though the prior work
reported very low communication cost for DBPG~\cite{LiAndSmoYu14},
we observe that a significant amount of
time was spent on data synchronization. The reason is twofold. First,
\cite{LiAndSmoYu14} pre-processed the data to remove tail features (tail
vertices in $V$) before training. But we fed the raw data into the algorithm and
let the $\ell_1$-regularizer do the feature selection automatically, which often
yields a better machine learning model but induces more network traffic. Second,
the network bandwidth of the university cluster we used is 20 times less than
the industrial data-center used by \cite{LiAndSmoYu14}.  Therefore, the
communication cost can not be ignored in our experiment. However, after the
partitioning, the inter-machine communication is decreased from 4.2TB to
0.3TB. Furthermore, the ratio of inner-machine traffic increases from 6\% to
92\%. In total, inter-machine communication is decreased by more than 90\%,
which significantly speeds inference.

\begin{table}[t]
  \centering
  \begin{tabular}{|l|rr|r|}
    \hline
    method & partition & inference & total \\
    \hline
    random & 0h         & 1.43h      & 1.43h  \\
    Parsa  & 0.07h      & 0.84h      & 0.91h  \\
    \hline
  \end{tabular}
  \caption{Time for $\ell_1$-regularized logistic
    regression on CTRb on 16
    machines requiring 45 data passes. }
  \label{tab:batch}
  \smallskip
  \begin{tabular}{|l|rr|r|}
    \hline
    method & inner-machine & inter-machine & total \\
    \hline
    random & 0.27           & 4.23          & 4.51  \\
    Parsa  & 3.68           & 0.32          & 4.00  \\
    \hline
  \end{tabular}
  \caption{Total data (TB) sent during inference.}
  \label{tab:batch_comm}
\end{table}

\section{Conclusion}
\label{sec:summary}

This paper presented a new parallel vertex-cut graph partition
algorithm, \algo, to solve the data and parameter placement problem. Our contributions are the following:
\begin{itemize}
\item We give theoretical analysis and approximation guarantees for
  both decomposition stages of what is generally an NP hard problem.
\item We show that the algorithm can be implemented very
  efficiently by judicious use of a doubly-linked list in $O(k|E|)$
  time.
\item We provide technologies such as sampling, initialization, and
  parallelizaiton, to improve the speed and partition quality.
\item Experiments show that \algo works well in practice,
  beating (or matching) all competing algorithms in \emph{both} memory
  footprint and communication cost while also offering very fast runtime.
\item We used \algo to accelerate a stat-of-the-art distributed solver for
  $\ell_1$-regularized logistic regression implemented in \ps. We observed a
  1.6x speedup on 16 machines with a dataset containing 10 billion nonzero
  entries.
\end{itemize}
In summary, \algo is a fast, relatively simple, highly scalable and
well performing algorithm.

\medskip

\noindent
{\bfseries Acknowledgments:} We thank Stefanie Jegelka and Christos
Faloutsos  for inspiring discussions.

\bibliographystyle{abbrv}

\end{document}